\documentclass[11pt,a4paper]{article}

\usepackage[margin=1in]{geometry}
\usepackage[T1]{fontenc}
\usepackage{lmodern}
\usepackage{amsmath,amssymb,amsthm}
\usepackage{graphicx}
\usepackage{microtype}
\usepackage[colorlinks=true,citecolor=blue,linkcolor=blue,urlcolor=blue]{hyperref}

\usepackage{bm}
\usepackage{pgfplots}
\usepackage{placeins}
\newtheorem{theorem}{Theorem}[section]
\newtheorem{lemma}[theorem]{Lemma}
\newtheorem{proposition}[theorem]{Proposition}
\newtheorem{corollary}[theorem]{Corollary}
\theoremstyle{definition}

\theoremstyle{remark}
\newtheorem{remark}[theorem]{Remark}
\pgfplotsset{compat=1.18}
\tikzset{
  jobL/.style={draw=black,fill=red!25},
  jobM/.style={draw=black,fill=blue!20},
  jobS/.style={draw=black,fill=green!25}
}

\newcommand{\opt}{\mathrm{OPT}}
\newcommand{\alg}{\mathrm{ALG}}
\newcommand{\rwspt}{\mathrm{RWSPT}}

\newcommand{\bp}{\bm{p}}
\newcommand{\bq}{\bm{q}}
\newcommand{\bmu}{\bm{\mu}}
\newcommand{\bw}{\bm{w}}
\newcommand{\bpi}{\bm{\pi}}
\newcommand{\Rpos}{\mathbb{R}_{>0}}

\title{The Exact Robustness of Smith's Rule with Predicted Processing Times}
\author{
  Ohad Ayalon\\
  Bar-Ilan University\\
  \href{mailto:ayalonoh@gmail.com}{\texttt{ayalonoh@gmail.com}}
  \and
  Ilan Reuven Cohen\\
  Bar-Ilan University\\
  \href{mailto:ilan-reuven.cohen@biu.ac.il}
       {\texttt{ilan-reuven.cohen@biu.ac.il}}
}
\date{}
\hypersetup{
  pdftitle={The Exact Robustness of Smith's Rule with Predicted Processing Times},
  pdfauthor={Ohad Ayalon and Ilan Reuven Cohen},
  pdfsubject={Scheduling algorithms with predicted processing times},
  pdfkeywords={Scheduling, Smith's rule, Algorithms with Prediction, approximation ratios, weighted completion time}
}

\begin{document}
\maketitle

\begin{abstract}
Smith's rule orders jobs by their weight-to-processing-time ratios.
We study the robustness of the resulting list schedule on identical
parallel machines when processing times are only approximately known.
Job \(j\) is ordered using a nominal processing time \(\mu_j\), while its
true processing time satisfies
\[
  \mu_j \leq p_j \leq \alpha\mu_j,
\]
where \(\alpha\geq1\) is the distortion parameter.

We determine the exact worst-case approximation ratio of the list schedule
induced by the nominal Smith order on identical parallel machines, measured
against the clairvoyant optimum and taking the supremum over the number of
machines. The ratio is an explicit continuous piecewise-algebraic function
of \(\alpha\). It equals the classical Kawaguchi--Kyan constant
\((1+\sqrt{2})/2\) at \(\alpha=1\), and becomes exactly
\(\sqrt{\alpha}\) for every \(\alpha\geq4\).
On a single machine, the exact ratio is \(\sqrt{\alpha}\).
For every fixed machine count, no deterministic nonpreemptive algorithm
can improve on this guarantee, even if it knows \(\alpha\).

The proof reduces to jobs with a common nominal Smith ratio and
\(p_j/\mu_j\in\{1,\alpha\}\). A prefix--suffix certificate, obtained by cutting the list schedule at
its minimum final machine load, bounds all earlier work together and each
remaining job piece separately. Two ways of packing a matching lower-bound
instance explain the algebraic branches below \(\alpha=4\).

We also show that any fixed schedule ordered by nominal Smith ratios on
each machine, with nominal approximation ratio \(1+\delta\), has realized
ratio at most \(\sqrt{\alpha}+\alpha\delta\). Combining this transfer
with a nominal approximation scheme gives, for every fixed
\(\varepsilon>0\), a deterministic polynomial-time algorithm, independent
of \(\alpha\), with ratio \((1+\varepsilon)\sqrt{\alpha}\) simultaneously
for every \(1\leq\alpha\leq4\). Thus computing a nominal assignment can
improve on the fixed list rule while retaining a uniform robustness guarantee.
\end{abstract}

\section{Introduction}

Smith's rule is a classical scheduling principle for minimizing total
weighted completion time. Consider a batch of jobs available at time zero,
each of which must run uninterrupted once it starts. Job $j$ has a
processing time $p_j>0$, a weight $w_j>0$, and a completion time $C_j$.
Interpreting $w_j$ as the cost per unit time while job $j$ remains
unfinished, the objective is the accumulated cost of unfinished jobs,
\[
    \sum_j w_j C_j.
\]
On one machine, ordering jobs in nonincreasing order of $w_j/p_j$
minimizes this objective exactly~\cite{smith1956various}. On identical
parallel machines, the same order defines the weighted shortest processing
time (WSPT) list rule: whenever a machine becomes idle, it starts the next
job in the list. The rule requires processing times before the jobs run.
How much does its performance deteriorate when these priorities are
computed from estimates?

Processing-time estimates may come from a heuristic or a learned
predictor. Prediction errors can then reverse the preferred order of
jobs. Processing-time predictions have therefore motivated
robust scheduling models that bound the effect of multiplicative
error~\cite{azar2021flow,azar2022distortion}. We study both the loss
incurred by the nominal Smith order and the guarantees obtainable by
computing a more careful assignment from the same estimates.

For a concrete setting, consider a batch of computational tasks ready to
run on identical workers. A duration estimate can guide the assignment,
and weights can express how costly it is to keep different tasks waiting
for their results. If tasks must run uninterrupted once dispatched, an
underestimated task can occupy a worker while other work waits. This
illustrative setting captures known weights, uncertain durations, and
uninterrupted execution without assuming a particular prediction method
or probability distribution.

\paragraph*{Classical parallel-machine benchmark.}
Even with exact processing times, the WSPT list rule can incur a loss on
parallel machines. Kawaguchi and Kyan proved that its best-possible
worst-case approximation ratio is
\[
    \rho_{\mathrm{KK}}:=\frac{1+\sqrt{2}}{2}
\]
when the processing times are known~\cite{kawaguchi1986worst}.
This constant is the starting point for studying how the guarantee changes
under prediction error.

\paragraph*{Prediction model.}
All \(n\) jobs and their weights are available initially. There are \(m\) identical
machines, so a job's duration does not depend on its assignment. Processing
is nonpreemptive: once started, a job runs to completion on its assigned
machine. Its true processing time is revealed only when it completes.
These assumptions isolate uncertainty in duration from the separate
effects of future arrivals, machine heterogeneity, and the ability to
interrupt a job. The algorithm receives a nominal (predicted) processing time
$\mu_j>0$ for every job and the true processing times satisfy
\[
    \mu_j\leq p_j\leq\alpha\mu_j
    \qquad\text{for every job }j,
\]
where $\alpha\geq 1$ is the distortion.  The robust weighted shortest processing
time rule, denoted by $\rwspt$, orders jobs in nonincreasing order of
$w_j/\mu_j$.  The $\rwspt$ algorithm need not know $\alpha$.  Whenever
nominal ratios are tied, our upper bounds include every valid tie-breaking;
\(V_L\), defined formally in
Section~\ref{sec:common-reductions}, denotes the largest resulting
list-schedule cost. On parallel machines the rule uses only the nominal
data and observed completions: start its first \(\min\{m,n\}\) jobs at time zero, and
whenever a machine becomes idle, dispatch the next unstarted job in the
order.  If several machines become idle simultaneously, ties among them may
be resolved arbitrarily.

Multiplicative error measures relative accuracy across jobs of different
sizes and is unchanged by the choice of time units. The common bound
\(\alpha\) allows different jobs to have different errors within the
interval; it imposes no distributional assumption. For \(\rwspt\), the
one-sided interval is only an analytical normalization: uniformly
rescaling arbitrary positive estimates preserves their Smith order and
reduces two-sided multiplicative errors to this form. The resulting
distortion is the product of the largest overestimate and underestimate,
as in robust flow-time scheduling~\cite{azar2022distortion};
see Remark~\ref{rem:one-sided-normalization}. This rescaling need not be
computed by the algorithm.

For a scheduling rule \(\mathcal A\), let
\(\mathcal A(\bmu,\bp,\bw,m)\) denote its total weighted completion time,
and let \(\opt(\bp,\bw,m)\) denote the value of an optimal clairvoyant
schedule.  We call \(\mathcal A\) \(f(\alpha)\)-robust if
\[
    \sup_{\substack{m,n\in\mathbb N\\
                    \bmu,\bp,\bw\in\Rpos^n\\
                    \bmu\leq\bp\leq\alpha\bmu}}
    \frac{\mathcal A(\bmu,\bp,\bw,m)}{\opt(\bp,\bw,m)}
    \leq f(\alpha),
\]
with the supremum restricted to \(m=1\) for the single-machine problem.
We compare costs on the same realized processing vector; only the
comparison schedule knows that vector before execution.

The supremum ranges over all nominal inputs and their allowed
realizations. Minimizing the largest absolute cost for a fixed nominal
input is a different objective: for a fixed assignment and job order on
each machine, completion times are nondecreasing in the job durations, so the
largest absolute cost occurs at \(\bp=\alpha\bmu\). The approximation
ratio also depends on the realization's optimum;
Appendix~\ref{app:two-job-example} illustrates this distinction.

\paragraph*{Our main results.}
We distinguish two algorithmic choices: executing one nominal priority
list, and using the nominal data to compute a fixed assignment and order
on each machine. Our first result determines the exact robustness of the
parallel-machine list rule. Let \(\mathcal R(\alpha)\) be its worst-case ratio, where the
supremum is over the number of jobs, the number of identical machines, and
all instances of distortion at most \(\alpha\).

\begin{theorem}[Parallel-machine main theorem]\label{thm:intro-parallel}
The worst-case approximation ratio of the \(\rwspt\) list rule is the
explicit function \(\mathcal L(\alpha)\) defined in
\eqref{eq:tight-parallel-ratio}:
\[
  \mathcal R(\alpha)=\mathcal L(\alpha).
\]
In particular, \(\mathcal L(\alpha)=\sqrt{\alpha}\) for
\(\alpha\geq4\).
The same supremum is obtained even if instances are required to have
distinct nominal Smith ratios.
\end{theorem}

Figure~\ref{fig:bounds} shows the exact curve. It starts at
\(\mathcal L(1)=(1+\sqrt2)/2\), and its two algebraic branches meet
continuously at \(\alpha_0=(17+\sqrt{33})/12\approx1.8954\).
Section~\ref{sec:exact-curve} specifies these branches exactly.
Sections~\ref{sec:lower-construction} and~\ref{sec:core-certificate}
give the lower-bound construction and the general upper-bound
certificate proof; Appendices~\ref{sec:tight-lower} and~\ref{sec:tight-upper}
supply the extremal parameters and scalar checks establishing equality.
The scheduling ideas behind these bounds are explained in
Section~\ref{sec:intro-techniques}.

The one-machine problem has a simpler exact answer and supplies the
\(\alpha\geq4\) branch of the parallel result.

\begin{theorem}[Single-machine theorem]\label{thm:intro-single}
For nonpreemptive weighted completion time on one machine, $\rwspt$ is
$\sqrt{\alpha}$-robust.  This is best possible for every deterministic
algorithm, even if the algorithm knows $\alpha$.
\end{theorem}

Lemma~\ref{lem:det-lb} strengthens this benchmark: for every fixed
number of machines \(m\), every deterministic nonpreemptive algorithm
has worst-case ratio at least \(\sqrt{\alpha}\), even knowing \(\alpha\).
The lower bound is over nominal inputs and allowed realizations for
that fixed \(m\); it does not assert optimality for each fixed nominal
input. Our second algorithmic result uses a nominal assignment to
approach this benchmark simultaneously throughout \(1\leq\alpha\leq4\).

\begin{theorem}[Near-optimal robust algorithm]
\label{thm:intro-robust-ptas}
For every fixed \(\varepsilon>0\), there is a deterministic polynomial-time
algorithm, independent of \(\alpha\), that takes the number \(m\) of
identical machines as part of the input and, from \((\bmu,\bw,m)\),
constructs a fixed nonpreemptive schedule \(S\) satisfying
\[
  \sum_jw_jC_j^S(\bp)
  \leq
  (1+\varepsilon)\sqrt{\alpha}\,
  \opt(\bp,\bw,m)
  \qquad
  \text{for every }1\leq\alpha\leq4
  \text{ and }\bmu\leq\bp\leq\alpha\bmu.
\]
\end{theorem}

The general nominal-to-robust transfer separates nominal approximation
quality from distortion; its proof combines Smith-ratio layers with a
hybrid machine assignment. Lemma~\ref{lem:nominal-robust-transfer} applies to every fixed
schedule that orders jobs by nominal Smith ratios on each machine and has
nominal cost at most \(1+\delta\) times the nominal optimum. Its realized
ratio is at most \(\sqrt{\alpha}+\alpha\delta\) for every \(\alpha\geq1\).
Corollary~\ref{cor:exact-nominal-optimum} sets \(\delta=0\): an exact
nominal optimum, Smith-ordered on each machine, achieves optimal
deterministic robustness \(\sqrt{\alpha}\) for every fixed \(m\),
simultaneously for all \(\alpha\geq1\). Computing it may be expensive.

Applying the polynomial-time approximation scheme (PTAS) of Skutella and
Woeginger to the nominal instance and choosing \(\delta=\varepsilon/2\)
gives the theorem above~\cite{skutella2000ptas}. The assignment is computed before execution and
does not depend on the realized errors. This separates the quality of the
nominal optimization from the loss caused by prediction error.

For \(\alpha\geq4\), the \(\rwspt\) list rule attains the benchmark
\(\sqrt{\alpha}\) exactly for every fixed \(m\). The nominal PTAS and the list rule
generally produce different schedules. It remains open whether, for every
fixed \(\varepsilon>0\), there is a deterministic polynomial-time algorithm,
independent of \(\alpha\) and with \(m\) as part of the input, that attains a
\((1+\varepsilon)\sqrt{\alpha}\) guarantee simultaneously for every
\(\alpha\geq1\).

\paragraph*{The classical problem through a distortion lens.}
The exact curve quantifies when the fixed priority list pays an additional
price for distributing jobs across machines. It describes how the
Kawaguchi--Kyan constant at
\(\alpha=1\) deforms as the processing-time information deteriorates,
until the deterministic obstruction \(\sqrt{\alpha}\) becomes decisive
at \(\alpha=4\). Below that threshold, the certificate explains the
extra list-rule loss and the change of extremal packing regime.
The general transfer shows how nominal assignments approach the
fixed-machine benchmark within any fixed factor \(1+\varepsilon\).

The connection also works in the opposite direction.  Combining our exact
one-machine theorem with the classical lower bound of Eastman, Even, and
Isaacs already recovers the tight \(\sqrt{\alpha}\) branch for
\(\alpha\geq4\).  Conversely, the upper-bound argument developed for the
distorted problem specializes at \(\alpha=1\) to a short certificate-based
derivation of the Kawaguchi--Kyan bound; see
Remark~\ref{rem:kawaguchi-kyan}. The same load decomposition therefore
accounts for both the classical constant and its change under
multiplicative error.

\begin{figure}[t]
\centering
\begin{tikzpicture}
\begin{axis}[
  width=0.86\textwidth,
  height=6.2cm,
  xmin=1,
  xmax=6,
  ymin=0.95,
  ymax=2.55,
  xlabel={distortion \(\alpha\)},
  ylabel={approximation ratio},
  grid=major,
  extra x ticks={1.8953802205},
  extra x tick labels={$\alpha_0$},
  extra x tick style={tick label style={font=\scriptsize,yshift=-8pt}},
  tick label style={font=\small},
  label style={font=\small},
  legend style={
    at={(0.02,0.98)},
    anchor=north west,
    draw=none,
    fill=white,
    fill opacity=0.9,
    text opacity=1,
    font=\small
  }
]
  \addplot[
    gray!75!black,
    densely dotted,
    thick,
    domain=1:6,
    samples=100
  ] {sqrt(x)};
  \addlegendentry{deterministic benchmark \(\sqrt{\alpha}\)}

  \addplot[
    orange!85!black,
    densely dashdotted,
    thick,
    domain=1:4,
    samples=100
  ] {min((1+sqrt(2))/2*x,1+sqrt(x)/2)};
  \addlegendentry{preliminary upper bound}

  \addplot[
    blue!70!black,
    very thick
  ]
  table[x=alpha,y=ratio,col sep=comma] {data/exact-ratio.csv};
  \addlegendentry{tight upper/lower bound \(\mathcal L(\alpha)\)}

  \addplot[
    blue!70!black,
    very thick,
    domain=4:6,
    samples=60
  ] {sqrt(x)};
  \addplot[black!45,dashed,thin,forget plot]
    coordinates {(1.8953802205,0.95) (1.8953802205,2.55)};
  \addplot[black!45,dashed,thin,forget plot]
    coordinates {(4,0.95) (4,2.55)};
\end{axis}
\end{tikzpicture}
\caption{The matching upper and lower bounds give
\(\mathcal R(\alpha)=\mathcal L(\alpha)\).  The dash-dotted curve is
the preliminary bound from Theorem~\ref{thm:parallel-baseline}.
The dotted curve is the deterministic benchmark for every fixed machine count.
Vertical lines mark \(\alpha_0\) and \(4\).
For \(\alpha\geq4\), the preliminary
and exact bounds both equal \(\sqrt{\alpha}\).
Below \(4\), the solid curve linearly interpolates generated samples of
the algebraic branches in \eqref{eq:tight-parallel-ratio}; the interpolation
illustrates the exact bound proved in the paper.}
\label{fig:bounds}
\end{figure}
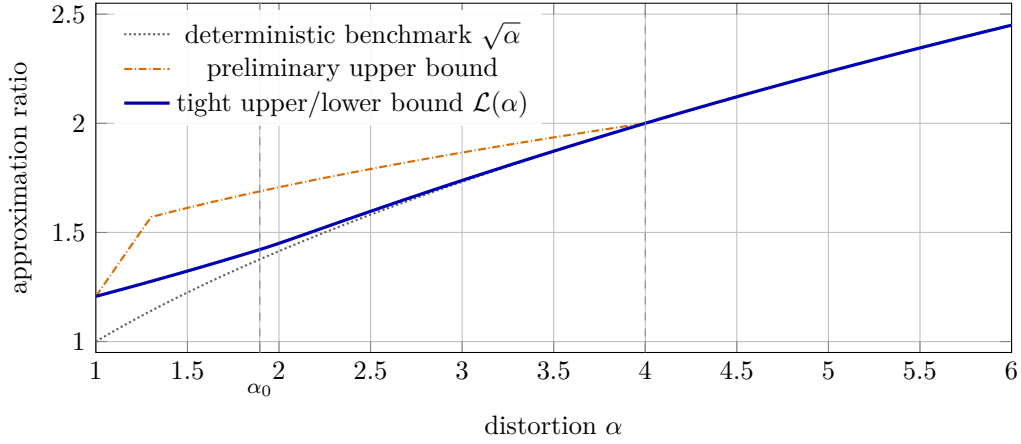

\subsection{Our techniques}
\label{sec:intro-techniques}

\paragraph*{Two structural reductions.}
The list-rule bounds begin with two reductions that isolate the effect of
prediction error. The equal-ratio lifting lemma shows that a uniform bound proved
when all jobs have the same nominal Smith ratio \(w_j/\mu_j\) automatically
extends to arbitrary nominal ratios.  We may then scale the common ratio so
that \(w_j=\mu_j\) for every job.  Next, fix the true processing times and
vary one pair \((\mu_j,w_j)\) while maintaining this normalization.  For a
fixed worst list schedule and a fixed optimal comparison schedule, both
costs vary linearly with the nominal time, so their ratio is maximized at
one end of the allowed interval. Repeating this for one job at a time
leaves only
\[
  \frac{p_j}{\mu_j}\in\{1,\alpha\}.
\]
We call these \emph{endpoint instances}; every job is either unstretched
or stretched by the full distortion factor.

\paragraph*{One machine.}
After the two reductions, the one-machine analysis has a simple
two-block interpretation.  All jobs have the same nominal Smith ratio,
so \(\rwspt\) may place the fully stretched jobs before the unstretched
ones.  Once the true processing times are known, Smith's rule places
these two blocks in the opposite order.  Thus the entire loss comes from
making the two blocks wait in the wrong order; comparing this
cross-block delay with the contributions internal to the blocks gives
the factor \(\sqrt{\alpha}\).  A matching adversarial construction
forces the same two-block reversal for every deterministic algorithm.
Consequently, \(\sqrt{\alpha}\) is exact even when the algorithm knows
\(\alpha\).

\paragraph*{A first parallel-machine consequence.}
For a fixed priority list, the start time of a job on \(m\) machines is at
most \(1/m\) times the total processing time of its predecessors.  The
single-machine theorem therefore bounds the waiting-time component of the
list schedule.  Combining this estimate with the lower bound of Eastman,
Even, and Isaacs gives
\[
  \max\left\{1+\frac{\sqrt{\alpha}}2,\sqrt{\alpha}\right\}.
\]
In particular, this already proves the tight \(\sqrt{\alpha}\) branch for
\(\alpha\geq4\).  Independently, monotonicity in the processing times and
the classical Kawaguchi--Kyan guarantee give
\(\alpha(1+\sqrt2)/2\); taking the better of the two estimates yields the
preliminary bound in Theorem~\ref{thm:parallel-baseline}.

\paragraph*{The tight parallel-machine obstruction.}
The examples used to prove the lower bound have three size classes, denoted by \(S\),
\(M\), and \(L\), all having the same nominal Smith ratio.  Normalize the
true size of every large \(L\)-job to one. The short jobs are fully
stretched, while medium and large jobs are unstretched. A valid
tie-breaking orders the short jobs first, the medium jobs second, and the
large jobs last. In the
resulting list schedule, every machine receives short and medium load, and
a fraction \(u\) of the machines receives no large job; each of the
remaining fraction \(1-u\) receives one.  Let \(x\) and \(z\) be the total
medium and short processing load, respectively, divided by the number
\(m\) of machines.
The comparison schedule isolates the large jobs on the \((1-u)m\)
machines and uses the other \(um\) machines for the remaining work.
For these lower-bound examples we choose \(x\leq u\), giving two regimes.
If \(x+z\leq u\), all medium and short
work fits on the machines without a large job.  If \(x\leq u<x+z\), the
short jobs first fill those machines to load one, and the residual short
load is spread over all \(m\) machines.  These two operations are
shown in Figure~\ref{fig:parallel-lower-cases}.

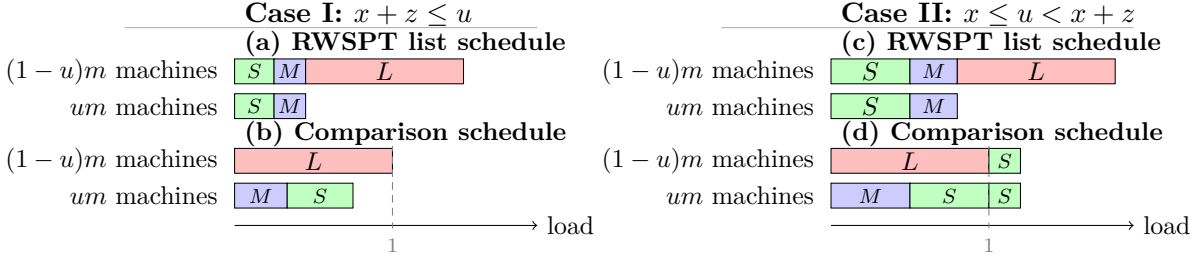
\begin{figure}[tbp]
\centering
\resizebox{\linewidth}{!}{%
\begin{tikzpicture}[
  x=1.2cm,
  y=0.68cm,
  font=\small,
  case/.style={anchor=west,font=\normalsize\bfseries},
  panel/.style={anchor=west,font=\small\bfseries},
  group/.style={anchor=east}
]
  \begin{scope}
    \node[case] at (0,3.75) {Case I: \(x+z\leq u\)};
    \draw[gray!65] (-1.25,3.48) -- (3.45,3.48);
    \node[panel] at (0,3.15) {(a) RWSPT list schedule};
    \node[group] at (-0.08,2.6) {\((1-u)m\) machines};
    \node[group] at (-0.08,1.9) {\(um\) machines};

    \draw[jobS] (0,2.35) rectangle ++(0.45,0.5);
    \node[font=\scriptsize] at (0.225,2.6) {\(S\)};
    \draw[jobM] (0.45,2.35) rectangle ++(0.36,0.5);
    \node[font=\scriptsize] at (0.63,2.6) {\(M\)};
    \draw[jobL] (0.81,2.35) rectangle ++(1.8,0.5);
    \node at (1.71,2.6) {\(L\)};

    \draw[jobS] (0,1.65) rectangle ++(0.45,0.5);
    \node[font=\scriptsize] at (0.225,1.9) {\(S\)};
    \draw[jobM] (0.45,1.65) rectangle ++(0.36,0.5);
    \node[font=\scriptsize] at (0.63,1.9) {\(M\)};

    \node[panel] at (0,1.35) {(b) Comparison schedule};
    \node[group] at (-0.08,0.8) {\((1-u)m\) machines};
    \node[group] at (-0.08,0.1) {\(um\) machines};

    \draw[jobL] (0,0.55) rectangle ++(1.8,0.5);
    \node at (0.9,0.8) {\(L\)};
    \draw[jobM] (0,-0.15) rectangle ++(0.6,0.5);
    \node[font=\scriptsize] at (0.3,0.1) {\(M\)};
    \draw[jobS] (0.6,-0.15) rectangle ++(0.75,0.5);
    \node[font=\scriptsize] at (0.975,0.1) {\(S\)};

    \draw[densely dashed,gray] (1.8,-0.3) -- (1.8,1.1);
    \draw[->] (0,-0.5) -- (3.45,-0.5) node[right] {load};
    \draw[gray] (1.8,-0.56) -- (1.8,-0.44);
    \node[gray,font=\scriptsize,below] at (1.8,-0.56) {\(1\)};
  \end{scope}

  \begin{scope}[shift={(6.8,0)}]
  \node[case] at (0,3.75) {Case II: \(x\leq u<x+z\)};
  \draw[gray!65] (-1.25,3.48) -- (3.45,3.48);
  \node[panel] at (0,3.15) {(c) RWSPT list schedule};
  \node[group] at (-0.08,2.6) {\((1-u)m\) machines};
  \node[group] at (-0.08,1.9) {\(um\) machines};

  \draw[jobS] (0,2.35) rectangle ++(0.9,0.5);
  \node at (0.45,2.6) {\(S\)};
  \draw[jobM] (0.9,2.35) rectangle ++(0.54,0.5);
  \node[font=\scriptsize] at (1.17,2.6) {\(M\)};
  \draw[jobL] (1.44,2.35) rectangle ++(1.8,0.5);
  \node at (2.34,2.6) {\(L\)};

  \draw[jobS] (0,1.65) rectangle ++(0.9,0.5);
  \node at (0.45,1.9) {\(S\)};
  \draw[jobM] (0.9,1.65) rectangle ++(0.54,0.5);
  \node[font=\scriptsize] at (1.17,1.9) {\(M\)};

  \node[panel] at (0,1.35) {(d) Comparison schedule};
  \node[group] at (-0.08,0.8) {\((1-u)m\) machines};
  \node[group] at (-0.08,0.1) {\(um\) machines};

  \draw[jobL] (0,0.55) rectangle ++(1.8,0.5);
  \node at (0.9,0.8) {\(L\)};
  \draw[jobS] (1.8,0.55) rectangle ++(0.36,0.5);
  \node[font=\scriptsize] at (1.98,0.8) {\(S\)};

  \draw[jobM] (0,-0.15) rectangle ++(0.9,0.5);
  \node[font=\scriptsize] at (0.45,0.1) {\(M\)};
  \draw[jobS] (0.9,-0.15) rectangle ++(0.9,0.5);
  \node[font=\scriptsize] at (1.35,0.1) {\(S\)};
  \draw[jobS] (1.8,-0.15) rectangle ++(0.36,0.5);
  \node[font=\scriptsize] at (1.98,0.1) {\(S\)};

  \draw[densely dashed,gray] (1.8,-0.3) -- (1.8,1.1);
  \draw[->] (0,-0.5) -- (3.45,-0.5) node[right] {load};
  \draw[gray] (1.8,-0.56) -- (1.8,-0.44);
  \node[gray,font=\scriptsize,below] at (1.8,-0.56) {\(1\)};
  \end{scope}
\end{tikzpicture}%
}
\caption{The two parallel-machine lower-bound regimes: Case I (left) and
Case II (right). The labels
\(S,M,L\) denote the short, medium, and large classes.  An \(S\)- or
\(M\)-block represents aggregate load from many jobs of that class,
whereas an \(L\)-block represents one large job, normalized to size \(1\).
In each list schedule, the \((1-u)m\) upper machines receive a large job
and the \(um\) lower machines do not.  In Case I the comparison schedule
keeps all remaining work on the latter machines.  In Case II it fills
them to load \(1\) and spreads the residual short load over every machine.
The aligned dashed lines mark load \(1\). Load equals elapsed time here;
widths are schematic.}
\label{fig:parallel-lower-cases}
\end{figure}

\paragraph*{From the bad example to the upper bound.}
If there are at most \(m\) jobs, all start at time zero and the list rule
is optimal. Otherwise, look at the smallest final load among the
list-schedule machines and scale this load to one. Cut every
machine at time one. Every machine has processed one unit of work before
the cut, and at most one job extends beyond it. Indeed, if a further job
started after the cut, the list rule would have dispatched it to a
machine that had already finished all its work. We call the work before
the cut the \emph{prefix}, and the remaining part of each last job a
\emph{suffix piece}.

For the analysis, split prefix jobs into arbitrarily small pieces, keeping
them in their list-schedule positions and dividing their nominal times
and weights in proportion to their processing times. This operation cannot
decrease the ratio of the inherited fragmented schedule's cost to the
fragmented instance's optimum. We then bound the total cost of the prefix and the
cost of each suffix piece separately. In the optimal schedule the same
pieces may move between machines, but they retain their prefix or suffix
classification from the list schedule.

The optimum's cost contains quadratic terms in the processing load placed
on each machine. We lower-bound these terms by a sum of two contributions:
one for all prefix pieces together and one for each suffix piece. The
contributions, multiplied by the proposed approximation ratio, must cover
the corresponding list-schedule costs. One inequality checks this for
every mixture of stretched and unstretched prefix work; the other checks
every possible suffix size. These two inequalities form what we call a
\emph{prefix--suffix certificate}. Figure~\ref{fig:prefix-suffix-certificate}
shows the pieces being compared. The functions and parameters used to
verify the inequalities are defined in Section~\ref{sec:core-certificate}.

\begin{figure}[tbp]
\centering
\resizebox{\linewidth}{!}{%
\begin{tikzpicture}[
  x=1.1cm,
  y=0.68cm,
  font=\small,
  certA/.style={draw=black,fill=blue!20},
  certB/.style={draw=black,fill=orange!25},
  certSuffix/.style={line width=1.15pt}
]
  \begin{scope}
  \node[anchor=west,font=\bfseries] at (-1.45,6.05)
    {(a) Normalized RWSPT list schedule};
  \node[anchor=south,align=center,font=\scriptsize] at (1.5,4.75)
    {aggregate\\prefix};
  \node[anchor=south,align=center,font=\scriptsize] at (3.45,4.75)
    {suffix\\pieces};

  \node[anchor=east] at (-0.12,4.2) {machine 1};
  \draw[certA] (0,3.96) rectangle (1.15,4.44);
  \draw[certB] (1.15,3.96) rectangle (2.05,4.44);
  \draw[certA] (2.05,3.96) rectangle (3,4.44);
  \draw[certA,certSuffix] (3,3.96) rectangle (3.72,4.44);

  \node[anchor=east] at (-0.12,3.4) {machine 2};
  \draw[certB] (0,3.16) rectangle (0.82,3.64);
  \draw[certA] (0.82,3.16) rectangle (1.92,3.64);
  \draw[certB] (1.92,3.16) rectangle (3,3.64);
  \draw[certB,certSuffix] (3,3.16) rectangle (3.86,3.64);

  \node[anchor=east] at (-0.12,2.8) {\(\vdots\)};
  \node[anchor=east] at (-0.12,2.2) {machine \(m\)};
  \draw[certA] (0,1.96) rectangle (0.72,2.44);
  \draw[certB] (0.72,1.96) rectangle (1.7,2.44);
  \draw[certA] (1.7,1.96) rectangle (3,2.44);

  \draw[densely dashed,gray] (3,1.77) -- (3,4.48);
  \node[gray,font=\scriptsize,anchor=north] at (3,1.77) {load \(1\)};
  \draw[->,gray!80!black] (1.4,1.86) -- (1.4,1.38);
  \node[draw=gray!70,rounded corners,align=center,text width=6.0cm,
        inner sep=5pt,anchor=north] at (1.4,1.2)
    {\textbf{List-schedule upper bound}\\
     one aggregate prefix cost\\
     + one cost for each suffix};

  \end{scope}

  \draw[->,gray!80!black] (4.2,3.4) -- (5.4,3.4)
    node[midway,above,font=\scriptsize]
    {same pieces};

  \begin{scope}[shift={(7.25,0)}]
  \node[anchor=west,font=\bfseries] at (-1.45,6.05)
    {(b) Optimal schedule};
  \node[anchor=east] at (-0.12,4.2) {machine 1};
  \draw[certA] (0,3.96) rectangle (0.8,4.44);
  \draw[certA,certSuffix] (0.8,3.96) rectangle (1.45,4.44);
  \draw[certA] (1.45,3.96) rectangle (1.92,4.44);
  \draw[certB] (1.92,3.96) rectangle (2.55,4.44);
  \draw[certB,certSuffix] (2.55,3.96) rectangle (3.37,4.44);

  \node[anchor=east] at (-0.12,3.4) {machine 2};
  \draw[certA] (0,3.16) rectangle (1.12,3.64);
  \draw[certB] (1.12,3.16) rectangle (3.02,3.64);

  \node[anchor=east] at (-0.12,2.8) {\(\vdots\)};
  \node[anchor=east] at (-0.12,2.2) {machine \(m\)};
  \draw[certA] (0,1.96) rectangle (1.42,2.44);
  \draw[certB] (1.42,1.96) rectangle (2.82,2.44);

  \draw[->,gray!80!black] (1.4,1.86) -- (1.4,1.38);
  \node[draw=gray!70,rounded corners,align=center,text width=6.0cm,
        inner sep=5pt,anchor=north] at (1.4,1.2)
    {\textbf{Optimal-schedule lower bound}\\
     one aggregate prefix contribution\\
     + one contribution for each suffix};

  \end{scope}

  \draw[certA] (1.0,-1.55) rectangle (1.25,-1.15);
  \node[anchor=west,font=\scriptsize] at (1.35,-1.35) {unstretched};
  \draw[certB] (4.4,-1.55) rectangle (4.65,-1.15);
  \node[anchor=west,font=\scriptsize] at (4.75,-1.35) {stretched};
  \draw[draw=black,line width=1.15pt] (7.3,-1.55) rectangle (7.58,-1.15);
  \node[anchor=west,font=\scriptsize] at (7.68,-1.35)
    {suffix piece};

  \node[draw=gray!70,rounded corners,inner sep=5pt,align=center]
    at (5.025,-2.35)
    {prefix inequality + one-job suffix inequality\\[-0.5mm]
     \(\Longrightarrow\quad\) global upper bound};
\end{tikzpicture}%
}
\caption{The prefix--suffix certificate: list schedule (left), optimal
schedule (right). Machine \(m\) has final load \(1\)
and determines the dashed cut in the normalized list schedule.  The cut
creates one aggregate prefix and at most one suffix piece per machine.  In
the optimal schedule the same pieces may be redistributed, but they retain
their list-schedule classification; thick outlines mark suffix pieces. The
optimal-cost lower bound has a contribution for the aggregate prefix and a
separate contribution for each suffix. Widths show processing load
schematically.}
\label{fig:prefix-suffix-certificate}
\end{figure}
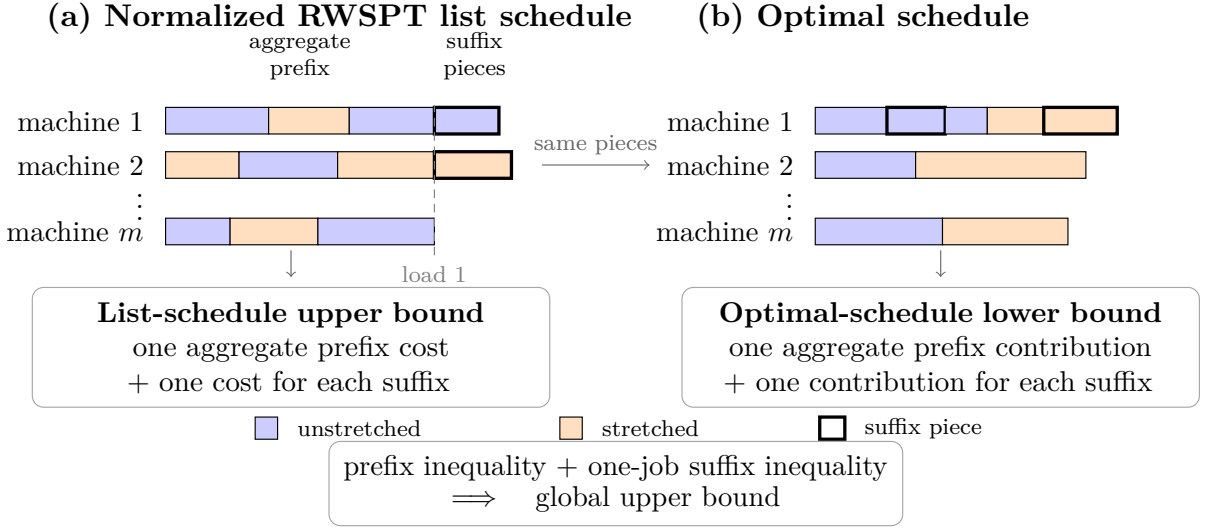

The three-class examples in Figure~\ref{fig:parallel-lower-cases} identify
the prefix mixture and suffix size at which these inequalities must be
tight. Their two packing regimes lead to
different equality cases for the suffix inequality, and hence to the two
algebraic branches. Checking the inequalities for every prefix mixture
and suffix size proves the upper bound for all endpoint instances; the
two initial reductions then cover arbitrary nominal ratios and errors.
The branch eliminations, Sturm calculation, and sign checks are collected
in Appendix~\ref{app:algebraic-certificates}. Finally, a vanishing weight
perturbation removes ties from the lower examples.

\subsection{Related work}
Eastman, Even, and Isaacs gave a fundamental lower bound on the optimal
parallel-machine cost in terms of the one-machine optimum and the
diagonal term \(\sum_jw_jp_j\)~\cite{eastman1964bounds}. Together with
their companion upper estimate for the list schedule, it yields the
classical \(3/2\) worst-case guarantee, subsequently sharpened to the
exact constant \(\rho_{\mathrm{KK}}\) by Kawaguchi and Kyan
\cite{kawaguchi1986worst}.

Schwiegelshohn gave a simpler alternative proof by reducing the analysis
to simple worst-case instances~\cite{schwiegelshohn2011alternative}.
J\"ager and Skutella refined this transformation framework to obtain the
exact deterministic WSPT ratio for every fixed number of machines and
improve the analysis of the stochastic weighted shortest expected
processing time (WSEPT) rule
\cite{jaeger2018generalizing}.
The stochastic setting can behave materially differently from the
deterministic one: for exponentially distributed processing times,
Jagtenberg, Schwiegelshohn, and Uetz show that the performance ratio of
WSEPT is at least \(1.243\), strictly above the deterministic
\(\rho_{\mathrm{KK}}\approx1.207\) bound~\cite{jagtenberg2013analysis}.
These stochastic guarantees compare expected cost with that of an optimal
nonanticipatory policy, which bases decisions only on information already
revealed. Our uncertainty is adversarial multiplicative error, and our
benchmark is the clairvoyant optimum for each realization.

Our \(\alpha=1\) argument shares the initial common-ratio reduction but
simplifies a different part of the proof.  It neither repeats the original
threshold analysis nor transforms the instance all the way to a canonical
worst case.  After one normalization, a prefix--suffix certificate directly
compares the list schedule with an arbitrary optimal schedule, and the
Kawaguchi--Kyan constant follows from two scalar inequalities.  We therefore
claim a new \emph{certificate-based derivation}, not the first simple proof
of the classical bound; see Remark~\ref{rem:kawaguchi-kyan}.  Its advantage
is reusability: the same certificate, with a second job type and one
additional reference load, proves the complete distorted upper bound.
Kashaev develops a semidefinite vector-fitting framework for selfish,
local, and online scheduling. Among other applications, it recovers the
Kawaguchi--Kyan constant as the pure price of anarchy of an identical-machine
scheduling game under Smith's rule: the worst ratio of a pure Nash
equilibrium's cost to the optimal assignment cost~\cite{kashaev2026vector}.
Our certificate is instead scalar, built directly from the normalized
list-schedule prefix and its at-most-one-job suffix structure, and tailored
to track multiplicative processing-time distortion exactly.

For computing an assignment rather than executing a priority list,
Skutella and Woeginger gave a PTAS for the general identical-machine
problem \(P\|\sum_jw_jC_j\), with the number of machines included in the
input~\cite{skutella2000ptas}.

Classical clairvoyant scheduling assumes that processing times are known when
decisions are made, whereas nonclairvoyant scheduling reveals a processing time
only when the job completes~\cite{motwani1994nonclairvoyant}.  A growing
learning-augmented literature interpolates between these models.  Purohit,
Svitkina, and Kumar and Im et al.\ study processing-time predictions for
preemptive single-machine total completion time
\cite{purohit2018improving,im2023non}.  Bampis et al.\ consider untrusted
predictions for total completion time with job releases on a single machine
and for batches on identical parallel machines
\cite{bampis2022untrusted}, while Zhao and Zomaya study smoothness
bounds for the single-machine problem~\cite{zhao2026smoothness}.  Lindermayr and
Megow use permutation predictions for weighted completion time on single,
identical, and unrelated machines; in the unrelated-machine case, the
prediction also specifies a job assignment~\cite{lindermayr2022permutation}.
More recent variants give predictions for only some jobs
\cite{benomar2024partial} or provide progress feedback during execution
\cite{benomar2025progress}. Several of these completion-time models permit preemption,
including models in which all jobs are available at time zero.
We study a nonpreemptive batch with known weights and multiplicative
processing-time error, and analyze both a fixed priority rule and a
computed nominal assignment.

The closest work in its error model is the online, preemptive robust
flow-time line initiated by Azar, Leonardi, and Touitou~\cite{azar2021flow}.
Flow time is the elapsed time from a job's arrival to its completion.
Their distortion-oblivious extension measures prediction quality by the
product of the largest overestimate and underestimate and gives
single-machine algorithms that need not know this distortion
\cite{azar2022distortion}. Later work treats multiple machines and total
stretch, the sum of jobs' flow times divided by their processing
times~\cite{azar2022multiple}; Gupta et al.\ obtain an asymptotically
optimal guarantee linear in the distortion for unweighted single-machine
flow time~\cite{gupta2026optimal}. We use the same multiplicative error
measure for nonpreemptive weighted completion time and determine the exact
worst-case ratio of the fixed \(\rwspt\) priority rule.

Prediction-error analyses also address makespan, the last completion time,
on uniform machines with different speeds~\cite{zhao2022uniform} and
identical machines, including nonpreemptive
batches~\cite{bampis2026makespan,chen2026tighter}.
Interval-uncertainty work studies both absolute and relative regret for
single-machine weighted completion time
\cite{pereira2016robust,sotskov2009minimizing}.
Relative regret already uses the realization-specific optimum, giving
the approximation ratio minus one. Our results give sharp dependence
on a common multiplicative distortion, including parallel-machine list
scheduling and a transfer for nominal assignments.
In stochastic queueing, inaccurate sizes affect size-based
policies~\cite{wierman2008scheduling}; the price of misprediction compares
mean response time under predicted and perfect information for the same
policy~\cite{mitzenmacher2020scheduling}; and two-sided multiplicative
estimates give uniform bounds on mean response time relative to clairvoyant
shortest remaining processing time (SRPT)~\cite{scully2022uniform}.
These stochastic models differ from our adversarial batch.
Appendix~\ref{app:extended-related-work} gives details;
see~\cite{mitzenmacher2022algorithms} for an overview of algorithms with predictions.

\subsection{Organization}
Section~\ref{sec:single} proves the shared structural reductions, the
single-machine bound, and the deterministic lower bound for every fixed
machine count. Section~\ref{sec:parallel-baseline} defines the exact
parallel curve, proves the preliminary bounds, constructs the three-class
lower-bound instances, and proves the general prefix--suffix certificate.
Section~\ref{sec:near-optimal-robustness} proves the nominal-to-robust
transfer and derives its algorithmic consequences.
Appendix~\ref{sec:parallel} establishes the branch properties.
Appendices~\ref{sec:tight-lower} and~\ref{sec:tight-upper} give the
extremal parameter choices, matching contacts, and scalar certificate
checks; Appendix~\ref{app:algebraic-certificates} records the supporting
algebra. The remaining appendices contain a numerical example and
further related work.

\section{Structural reductions and the single-machine bound}
\label{sec:single}

\subsection{Common notation and reductions}
\label{sec:common-reductions}

Let \(\Pi(\bmu,\bw)\) be the set of orders that are nonincreasing in
\(w_j/\mu_j\).  For an order
\(\bpi=(\pi_1,\ldots,\pi_n)\), let
\[
  V_m(\bp,\bw,\bpi)
  :=\sum_j w_j C_j^{\operatorname{LS}_m(\bpi)}
\]
be the cost of its list schedule on \(m\) identical machines.  Operationally,
\(\operatorname{LS}_m(\bpi)\) starts the first \(\min\{m,n\}\) jobs at time
zero and dispatches the next unstarted job in \(\bpi\) whenever a machine
becomes idle.  This implementation uses only completion events and is
therefore nonclairvoyant.  Ties among machines that become idle simultaneously
may be resolved arbitrarily.  For every fixed realized processing vector
\(\bp\), the resulting schedule is equivalent to the usual static
representation that assigns the jobs successively in order \(\bpi\) to a
machine of minimum accumulated load.  We use this retrospective static
representation in the proofs.

Let
\(\opt(\bp,\bw,m)\) be the minimum weighted completion time on those
machines.  A robust instance consists of \(m,n\in\mathbb N\) and vectors
\(\bmu,\bp,\bw\in\Rpos^n\), and is written
\(I=(\bmu,\bp,\bw;m)\). Vector inequalities are understood
componentwise. Define
\[
  V_L(I):=\max_{\bpi\in\Pi(\bmu,\bw)}V_m(\bp,\bw,\bpi),
  \qquad
  \opt(I):=\opt(\bp,\bw,m),
\]
and write \(V_L(\bmu,\bp,\bw,m):=V_L(I)\).  Thus \(V_L\) uses the most
expensive valid tie-breaking, so an upper bound on \(V_L\) applies to
every implementation of \(\rwspt\).

Before turning to the structural reductions, we record the normalization
promised in the introduction.

\begin{remark}[One-sided normalization]
\label{rem:one-sided-normalization}
Suppose the algorithm receives arbitrary estimates \(e_j>0\), and define
\[
  u:=\max_j\frac{p_j}{e_j},
  \qquad
  o:=\max_j\frac{e_j}{p_j}.
\]
Setting \(\mu_j:=e_j/o\) gives
\(\mu_j\leq p_j\leq (uo)\mu_j\).
Moreover, uniformly rescaling the estimates does not change their Smith
order, since \(w_j/\mu_j=o\,w_j/e_j\).  Thus, for \(\rwspt\), the
two-sided error model is equivalent to the one-sided model with distortion
\(\alpha=uo\).
\end{remark}

With this convention in place, we use two reductions for both the single-
and parallel-machine
analyses.  The first shows that it is enough to establish a uniform
bound when all jobs have the same nominal Smith ratio.

\begin{lemma}[Equal-ratio lifting]
\label{lem:equal-ratio-lifting}
Fix \(m\in\mathbb N\), \(\alpha\geq1\), and \(\gamma\geq1\).  Suppose
that
\(V_L(J)\leq\gamma\opt(J)\)
for every positive instance \(J=(\bmu',\bp',\bw';m)\) satisfying
\(\bmu'\leq\bp'\leq\alpha\bmu'\) and having a common nominal Smith
ratio.  Then \(V_L(I)\leq\gamma\opt(I)\) for every positive instance
\(I=(\bmu,\bp,\bw;m)\) satisfying
\(\bmu\leq\bp\leq\alpha\bmu\).
\end{lemma}

\begin{proof}
We induct on the number
\[
  \nu(I):=
  \left|
    \left\{\frac{w_j}{\mu_j}:j\in[n]\right\}
  \right|
\]
of distinct nominal Smith ratios.  The case \(\nu(I)=1\) is the
hypothesis.

Suppose \(\nu(I)>1\).  Choose a priority list attaining \(V_L(I)\), let
\(S\) be its list schedule, and let \(S^*\) be an optimal schedule.
Let \(H\subsetneq[n]\) be the jobs of maximum nominal Smith ratio, and
write
\[
  r_1:=\frac{w_j}{\mu_j}\quad(j\in H),
  \qquad
  r_2:=\max_{j\notin H}\frac{w_j}{\mu_j},
  \qquad
  \varepsilon:=\frac{r_2}{r_1}\in(0,1).
\]
Decompose the two objective values according to \(H\):
\[
  x:=\sum_{j\in H}w_jC_j^S,
  \quad
  y:=\sum_{j\notin H}w_jC_j^S,
  \qquad
  x^*:=\sum_{j\in H}w_jC_j^{S^*},
  \quad
  y^*:=\sum_{j\notin H}w_jC_j^{S^*}.
\]
Thus \(V_L(I)=x+y\) and \(\opt(I)=x^*+y^*\).

Let \(I_H\) be the instance consisting only of the jobs in \(H\), on
the same \(m\) machines.  Every job in \(H\) strictly precedes every job
outside \(H\) in an \(\rwspt\) priority list, so the restriction of
\(S\) to \(H\) is a list schedule for \(I_H\).  Deleting the other jobs
from \(S^*\) can only decrease the completion times of the jobs in
\(H\).  The common-ratio hypothesis therefore gives
\begin{equation}\label{eq:sosa-top-class}
  x\leq V_L(I_H)
  \leq\gamma\opt(I_H)
  \leq\gamma x^*.
\end{equation}

Now form \(\widetilde I=(\bmu,\bp,\widetilde\bw;m)\) by setting
\[
  \widetilde w_j:=
  \begin{cases}
    \varepsilon w_j,&j\in H,\\
    w_j,&j\notin H.
  \end{cases}
\]
This preserves all distortion constraints and merges the ratio of
every job in \(H\) with the next ratio class, so
\(\nu(\widetilde I)=\nu(I)-1\). Every previously valid priority list
remains valid for \(\widetilde I\); merging classes can only add tied
orders. In particular, the list generating \(S\) produces the same schedule.
The inductive hypothesis and the feasibility of \(S^*\) give
\begin{equation}\label{eq:sosa-merged-class}
  \varepsilon x+y
  \leq V_L(\widetilde I)
  \leq\gamma\opt(\widetilde I)
  \leq\gamma(\varepsilon x^*+y^*).
\end{equation}

Multiplying \eqref{eq:sosa-top-class} by \(1-\varepsilon>0\) and adding
\eqref{eq:sosa-merged-class} gives
\[
  x+y=(1-\varepsilon)x+(\varepsilon x+y)
  \leq\gamma(x^*+y^*).
\]
This completes the induction.
\end{proof}

After equal-ratio lifting, scaling all weights by a common factor lets
us normalize \(w_j=\mu_j\).  We call such an instance an
\emph{endpoint instance} if additionally
\[
  \frac{p_j}{\mu_j}\in\{1,\alpha\}
  \qquad\text{for every job }j.
\]
For \(\alpha>1\), the jobs split into the unstretched type \(A\) and
the fully stretched type \(B\):
\[
  A:=\{j:p_j=\mu_j\},\qquad
  B:=\{j:p_j=\alpha\mu_j\}.
\]
These sets are disjoint and cover all jobs. When \(\alpha=1\), the
two endpoint conditions coincide.

\begin{lemma}[Endpoint reduction]
\label{lem:endpoint-reduction}
Fix \(m,n\in\mathbb N\) and \(\alpha\geq1\).  Let
\(I=(\bmu,\bp,\bw;m)\) be a positive common-ratio instance satisfying
\(\bmu\leq\bp\leq\alpha\bmu\).  Then there is an endpoint instance
\(I'=(\bmu',\bp,\bw';m)\) with the same true processing vector and
\[
  \frac{V_L(I')}{\opt(I')}
  \geq
  \frac{V_L(I)}{\opt(I)}.
\]
\end{lemma}

\begin{proof}
The claim is immediate when \(\alpha=1\), so assume \(\alpha>1\).
Let \(r>0\) be the common value of \(w_j/\mu_j\).  Replacing every
weight \(w_j\) by \(w_j/r\) scales both \(V_L\) and \(\opt\) by
\(1/r\), does not change the valid priority lists, and gives
\(w_j=\mu_j\).  We may therefore assume this normalization.

Fix a job \(j\) and set
\[
  x_0:=\mu_j,
  \qquad
  x_-:=\frac{p_j}{\alpha},
  \qquad
  x_+:=p_j.
\]
For \(x\in[x_-,x_+]\), let \(I(x)\) be obtained by replacing both
\(\mu_j\) and \(w_j\) by \(x\), leaving all other coordinates fixed.
Throughout this interval, \(x\leq p_j\leq\alpha x\), and every nominal
Smith ratio remains one. Thus every \(I(x)\) is a valid common-ratio
instance with the same set of valid priority lists. The two endpoints
satisfy \(p_j/x_-=\alpha\) and \(p_j/x_+=1\).

Choose a priority list \(\bpi_L\) attaining \(V_L(I(x_0))\), and an
optimal schedule \(S^*\) for \(I(x_0)\).  Since \(\bp\) is fixed, the
completion times in these two schedules do not depend on \(x\).
Consequently, for constants \(A_L,A_*,b_L,b_*\geq0\),
\[
  V_m(\bp,\bw(x),\bpi_L)=A_L+b_Lx,
  \qquad
  \sum_i w_i(x)C_i^{S^*}=A_*+b_*x.
\]
The ratio
\(F(x):=\frac{A_L+b_Lx}{A_*+b_*x}\)
is monotone on \([x_-,x_+]\), because its derivative has constant sign.
Thus some \(x_1\in\{x_-,x_+\}\) satisfies \(F(x_1)\geq F(x_0)\).
Both comparison schedules remain feasible at \(x_1\), and therefore
\[
  \frac{V_L(I(x_1))}{\opt(I(x_1))}
  \geq F(x_1)
  \geq F(x_0)
  =
  \frac{V_L(I(x_0))}{\opt(I(x_0))}.
\]
Applying this operation successively to every job gives the claimed
endpoint instance.
\end{proof}

\subsection{The exact single-machine bound}

For \(m=1\), write \(V(\bp,\bw,\bpi):=V_1(\bp,\bw,\bpi)\) and
\(\opt(\bp,\bw):=\opt(\bp,\bw,1)\).  Explicitly,
\[
  V(\bp,\bw,\bpi)
  =\sum_{k=1}^n w_{\pi_k}\sum_{\ell=1}^k p_{\pi_\ell}.
\]
On one machine, \(\rwspt\) may return any
\(\bpi\in\Pi(\bmu,\bw)\), and \(V_L\) includes every such tie-breaking.

We use the following standard form of Smith's rule.

\begin{lemma}[Smith's rule~\cite{smith1956various}]\label{lem:smith-sosa}
For positive processing times and weights, an order minimizes
$V(\bp,\bw,\bpi)$ if and only if it is nonincreasing in $w_j/p_j$, up to
ties.  An order nondecreasing in $w_j/p_j$ maximizes $V(\bp,\bw,\bpi)$.
\end{lemma}

\begin{proof}
Consider two adjacent jobs $i$ and $j$.  The cost of placing $i$ before $j$
minus the cost of placing $j$ before $i$ is
\(p_iw_j-p_jw_i\).
Hence the first order is no more expensive exactly when
$w_i/p_i\geq w_j/p_j$.  Repeated adjacent exchanges prove both statements.
\end{proof}

The shared reductions now leave only a simple ordering choice on one
machine.

\begin{proposition}[Canonical one-machine instances]
\label{prop:canonical}
Fix \(\alpha\geq1\).  For the one-machine upper bound, it suffices to
consider instances that, after relabeling, satisfy:
\begin{enumerate}
  \item the algorithmic order is $(1,2,\ldots,n)$ and $w_j=\mu_j$ for
        every $j$;
  \item $p_j/\mu_j\in\{1,\alpha\}$ for every $j$;
  \item for some $k\in\{0,1,\ldots,n\}$,
        $p_j=\alpha w_j$ for $j\leq k$ and $p_j=w_j$ for $j>k$.
\end{enumerate}
\end{proposition}

\begin{proof}
Lemma~\ref{lem:equal-ratio-lifting} reduces the upper bound to
common-ratio instances, and
Lemma~\ref{lem:endpoint-reduction} then reduces them to
endpoint instances.  Thus \(w_j=\mu_j\) and all nominal Smith ratios are
tied, so every order is a valid \(\rwspt\) order.  By
Lemma~\ref{lem:smith-sosa}, the largest algorithmic cost is attained by
a nondecreasing order of the true ratios \(w_j/p_j\).
The jobs with $p_j=\alpha w_j$ therefore form a prefix, followed by the
jobs with \(p_j=w_j\).  Relabeling this order gives the three properties.
\end{proof}

The two endpoint types give a stretched block followed by an unstretched
block; the optimal true Smith order reverses them.

The canonical instances admit a short exact calculation.  For a
nonnegative vector $\bm{x}=(x_1,\ldots,x_s)$, including the empty vector,
define
\[
  T(\bm{x})=\sum_{i=1}^s x_i,
  \qquad
  S(\bm{x})=\sum_{i=1}^s x_i^2,
  \qquad
  R(\bm{x})=\frac{T(\bm{x})^2+S(\bm{x})}{2}.
\]
If a collection of jobs has both weights and processing times equal to
$\bm{x}$, then its weighted completion time, in any order, is
$R(\bm{x})$: each diagonal product $x_i^2$ appears once, and each product
$x_ix_j$ for $i\neq j$ appears once for the unordered pair $\{i,j\}$.

\begin{lemma}[Canonical bound]\label{lem:canonical-bound}
Every instance satisfying the three properties in
Proposition~\ref{prop:canonical} has approximation ratio at most
$\sqrt{\alpha}$.
\end{lemma}

\begin{proof}
Let $\bm{x}=(w_1,\ldots,w_k)$ be the weights of the stretched prefix and
let $\bm{y}=(w_{k+1},\ldots,w_n)$ be the remaining weights.  Set
\[
    r=R(\bm{x}),\qquad s=R(\bm{y}),\qquad
    z=T(\bm{x})T(\bm{y}).
\]
The algorithm places the stretched jobs first, whereas an optimal Smith
order places the unstretched jobs first.  Write \(\alg\) and \(\opt\) for
the costs of these two orders, respectively.  Consequently,
\begin{align*}
  \alg &= \alpha r+s+\alpha z,\\
  \opt &= \alpha r+s+z.
\end{align*}
Let $q=\sqrt{\alpha}$.  Since $T(\bm{x})^2\leq 2R(\bm{x})$ and likewise
for $\bm{y}$,
\(z\leq 2\sqrt{rs}\).
It follows that
\begin{align*}
 q\opt-\alg
 &= (q-1)\bigl(q^2r+s-qz\bigr)\\
 &\geq (q-1)\bigl(q^2r+s-2q\sqrt{rs}\bigr)\\
 &= (q-1)(q\sqrt r-\sqrt s)^2
 \geq 0.
\end{align*}
Thus $\alg/\opt\leq q=\sqrt{\alpha}$.
\end{proof}

\begin{proof}[Proof of the upper bound in Theorem~\ref{thm:intro-single}]
Proposition~\ref{prop:canonical} reduces the problem to canonical
instances, and Lemma~\ref{lem:canonical-bound} bounds each such ratio by
\(\sqrt{\alpha}\).
\end{proof}

We finish the single-machine analysis with a lower bound that holds for
every fixed number of machines. The algorithm may idle and choose later
jobs using all observations made during execution.

\begin{lemma}[Deterministic lower bound]\label{lem:det-lb}
For every fixed \(m\geq1\) and \(\alpha\geq1\), every deterministic
nonpreemptive algorithm on \(m\) identical machines has worst-case
approximation ratio at least \(\sqrt{\alpha}\), even if it knows \(\alpha\).
\end{lemma}

\begin{proof}
Fix \(m\geq1\) and a deterministic nonpreemptive algorithm. The claim is
immediate for \(\alpha=1\), so assume \(\alpha>1\) and put
\(q=\sqrt{\alpha}\). For each integer \(N\geq1\), set
\(K=\lfloor qN\rfloor\) and give the algorithm \(m(N+K)\) jobs, all
with \(w_j=\mu_j=1\).

Simulate the algorithm, assigning processing time \(\alpha\) to the first
\(mN\) jobs it starts and processing time \(1\) to every later job.
Order simultaneous starts deterministically, for example by machine
identifier. Processing times are revealed only at completion.
The simulation defines one fixed vector
\(\bp\in[1,\alpha]^{m(N+K)}\). On replay with this vector, the nominal
data, observed completions, elapsed times, and running-job histories are
the same, so determinism reproduces every start and every idle interval.
This remains true when the algorithm knows \(\alpha\) or uses the elapsed
processing time of unfinished jobs. If it never completes all jobs,
assign duration \(1\) to any unstarted jobs; its cost is infinite on this
fixed instance. We may therefore assume all jobs complete.

Let \(T\) be the first short-job start time. By \(T\), the first \(mN+1\)
starts have occurred, including simultaneous starts in the selected
order. Some machine has started at least \(N+1\) jobs. Its first \(N\)
jobs are long, and their nonpreemptive execution requires at least
\(\alpha N\) time before its next start. Thus \(T\geq\alpha N\).

Let \(\ell_i\) and \(s_i\) be the numbers of long and short jobs executed
on machine \(i\), respectively. Then
\(\sum_i\ell_i=mN\) and \(\sum_i s_i=mK\).
Write \(\operatorname{tri}(x):=x(x+1)/2\).
The \(r\)th long job on a machine completes no earlier than
\(\alpha r\). Moreover,
\[
 \sum_{i=1}^m\operatorname{tri}(\ell_i)
 -m\operatorname{tri}(N)
 =\frac12\sum_{i=1}^m(\ell_i-N)^2\geq0.
\]
Hence the long jobs contribute at least
\(\alpha m\operatorname{tri}(N)\) to the algorithmic cost.
Every short job starts at or after \(T\), so the \(r\)th short completion
on a machine is at or after \(T+r\). Applying the same balancing identity
to the \(s_i\)'s, their total contribution is at least
\[
 mKT+\sum_{i=1}^m\operatorname{tri}(s_i)
 \geq \alpha mNK+m\operatorname{tri}(K).
\]
These inequalities permit arbitrary idle time; no work-conserving
assumption is needed. Consequently
\[
 \alg\geq
 m\bigl[\alpha\operatorname{tri}(N)+\operatorname{tri}(K)+\alpha NK\bigr].
\]

A feasible clairvoyant schedule puts \(K\) short jobs followed by \(N\)
long jobs on each machine, with no idle time. We use only its feasibility,
not its optimality. Thus
\begin{align*}
 \opt&\leq
 m\bigl[\alpha\operatorname{tri}(N)+\operatorname{tri}(K)+NK\bigr],\\
 \frac{\alg}{\opt}&\geq
 \frac{\alpha\operatorname{tri}(N)+\operatorname{tri}(K)+\alpha NK}
      {\alpha\operatorname{tri}(N)+\operatorname{tri}(K)+NK}.
\end{align*}
As \(N\to\infty\), \(K/N\to q\). Dividing the last numerator and
denominator by \(N^2\) shows that the finite lower-bound expression tends
to
\[
 \frac{q^2/2+q^2/2+q^3}{q^2/2+q^2/2+q}
 =q.
\]
For this fixed \(m\), the worst-case approximation ratio is a supremum
over all finite inputs and allowed realizations, so it is at least
\(\sqrt{\alpha}\).
\end{proof}

\begin{proof}[Proof of Theorem~\ref{thm:intro-single}]
The upper bound was proved above, and Lemma~\ref{lem:det-lb} gives the
matching lower bound for every deterministic algorithm.
\end{proof}

\section{Exact robustness on parallel machines}
\label{sec:parallel-baseline}

\subsection{A preliminary bound}
Monotonicity and the classical WSPT bound give one preliminary estimate;
the single-machine result and the Eastman--Even--Isaacs bound give another.

\begin{theorem}[Baseline parallel-machine bound]\label{thm:parallel-baseline}
Let $\rho_{\mathrm{KK}}:=(1+\sqrt2)/2$.  For every \(m,n\in\mathbb N\),
\(\alpha\geq1\), and \(\bmu,\bp,\bw\in\Rpos^n\) satisfying
\(\bmu\leq\bp\leq\alpha\bmu\),
\[
 \frac{V_L(\bmu,\bp,\bw,m)}{\opt(\bp,\bw,m)}
 \leq
 \min\left\{
   \rho_{\mathrm{KK}}\alpha,
   \max\left\{1+\frac{\sqrt\alpha}{2},\sqrt\alpha\right\}
 \right\}.
\]
\end{theorem}

\begin{lemma}[Monotonicity of list scheduling]\label{lem:list-monotone}
Fix \(m\in\mathbb N\), a priority order
\(\bpi=(\pi_1,\ldots,\pi_n)\), and two processing vectors
\(\bp\leq\bp'\).  If \(C_j\) and \(C'_j\) are the completion times in
the corresponding list schedules on the same \(m\) machines, then
\[
  C_j\leq C'_j\qquad\text{for every job }j.
\]
\end{lemma}

\begin{proof}
For \(k=0,\ldots,n\), let
\[
  x_1^{(k)}\leq\cdots\leq x_m^{(k)}
  \quad\text{and}\quad
  y_1^{(k)}\leq\cdots\leq y_m^{(k)}
\]
be the sorted machine-load vectors after the first \(k\) jobs under
\(\bp\) and \(\bp'\), respectively.  We prove by induction that
\(x_i^{(k)}\leq y_i^{(k)}\) for every \(i\).  This is immediate for
\(k=0\).

Suppose it holds after \(k-1\) jobs.  The next job \(\pi_k\) is assigned
to a least-loaded machine, so
\[
  C_{\pi_k}=x_1^{(k-1)}+p_{\pi_k}
  \leq
  y_1^{(k-1)}+p'_{\pi_k}
  =C'_{\pi_k}.
\]
The new sorted load vectors are obtained by sorting
\[
 \bigl(x_1^{(k-1)}+p_{\pi_k},x_2^{(k-1)},\ldots,x_m^{(k-1)}\bigr)
 \quad\text{and}\quad
 \bigl(y_1^{(k-1)}+p'_{\pi_k},y_2^{(k-1)},\ldots,y_m^{(k-1)}\bigr).
\]
The first unsorted vector is coordinatewise no larger than the second,
and sorting preserves this order between their order statistics.  Hence
\(x_i^{(k)}\leq y_i^{(k)}\) for every \(i\), completing the induction.
\end{proof}

\begin{proof}[First bound in Theorem~\ref{thm:parallel-baseline}]
Fix any \(\bpi\in\Pi(\bmu,\bw)\).  Lemma~\ref{lem:list-monotone},
homogeneity, and $\bp\leq\alpha\bmu$ give
\[
 V_m(\bp,\bw,\bpi)
 \leq
 V_m(\alpha\bmu,\bw,\bpi)
 =\alpha V_m(\bmu,\bw,\bpi).
\]
For processing vector $\bmu$, this is ordinary WSPT list scheduling, whose
approximation ratio is $\rho_{\mathrm{KK}}$~\cite{kawaguchi1986worst}.
Moreover,
$\opt(\bmu,\bw,m)\leq\opt(\bp,\bw,m)$ because the optimum is monotone in
the processing times.  Hence
\[
 V_m(\bp,\bw,\bpi)
 \leq \alpha\rho_{\mathrm{KK}}\,\opt(\bmu,\bw,m)
 \leq \alpha\rho_{\mathrm{KK}}\,\opt(\bp,\bw,m).
\]
Maximizing over \(\bpi\in\Pi(\bmu,\bw)\) proves the first bound.
\end{proof}

We next prove the distortion-sensitive bound inherited from the
single-machine theorem.  Define
\[
   D:=\sum_j w_jp_j=\opt(\bp,\bw,n),
   \qquad
   C_1^*:=\opt(\bp,\bw,1).
\]
Eastman, Even, and Isaacs~\cite{eastman1964bounds} give the lower bound
\begin{equation}\label{eq:eastman-lb}
 \opt(\bp,\bw,m)
 \geq
 \max\left\{
   \frac{C_1^*}{m}+\frac{m-1}{2m}D,
   D
 \right\}.
\end{equation}

\begin{proof}[Second bound in Theorem~\ref{thm:parallel-baseline}]
Fix any $\bpi=(\pi_1,\ldots,\pi_n)\in\Pi(\bmu,\bw)$.  Write
\(C_{\pi_k}^{(m)}\) for the completion time of job \(\pi_k\) in list
scheduling on \(m\) machines and \(C_{\pi_k}^{(1)}\) for its completion time in the same
order on one machine.  The start time of job $\pi_k$ on $m$ machines is at
most
\[
  \frac{1}{m}\sum_{\ell<k}p_{\pi_\ell}
  =\frac{C_{\pi_k}^{(1)}-p_{\pi_k}}{m}.
\]
Indeed, the assertion is immediate for the first $m$ jobs; afterward, all
machines are busy before the next job starts, and the inequality follows
by comparing processed volume with the total size of earlier jobs.
Therefore
\begin{align}
 V_m(\bp,\bw,\bpi)
 &\leq
 \frac{V(\bp,\bw,\bpi)}{m}+\frac{m-1}{m}D \notag\\
 &\leq
 \frac{\sqrt\alpha}{m}C_1^*+\frac{m-1}{m}D,
 \label{eq:parallel-alg-upper}
\end{align}
where the last inequality is Theorem~\ref{thm:intro-single}.

Let \(q=\sqrt\alpha\) and \(\tau=C_1^*/D\geq1\).  Combining
\eqref{eq:eastman-lb} and \eqref{eq:parallel-alg-upper} yields
\[
 \frac{V_m(\bp,\bw,\bpi)}{\opt(\bp,\bw,m)}
 \leq
 H(\tau)
 :=
 \frac{q\tau+m-1}
      {\max\{\tau+(m-1)/2,m\}}.
\]
The two terms in the denominator agree at
$\tau_0=(m+1)/2$.  On $[1,\tau_0]$, $H$ is increasing.  On
$[\tau_0,\infty)$, the derivative has the sign of
\((m-1)\left(\frac q2-1\right)\).
If $q\leq2$, the maximum is attained at $\tau_0$, and
\[
 H(\tau_0)
 =1+\frac q2+\frac{q-2}{2m}
 \leq1+\frac q2.
\]
If $q\geq2$, the supremum is $q$; the value at $\tau_0$ is also at most
$q$.  Thus
\[
 H(\tau)\leq
 \max\left\{1+\frac{\sqrt\alpha}{2},\sqrt\alpha\right\}.
\]
This holds for every valid priority list, so maximizing over
\(\bpi\in\Pi(\bmu,\bw)\) gives the second bound.  Taking the better of the
two bounds proves the theorem.
\end{proof}

Together with Lemma~\ref{lem:det-lb}, this proves the exact
\(\sqrt{\alpha}\) ratio for every fixed \(m\) when \(\alpha\geq4\).

\subsection{The exact curve}
\label{sec:exact-curve}
The two packing configurations in Figure~\ref{fig:parallel-lower-cases}
yield the following precise curve for Theorem~\ref{thm:intro-parallel}.
Let
\[
  s_0:=\frac{1+\sqrt{33}}4,
  \qquad
  s_+:=1+\frac1{\sqrt2},
  \qquad
  \alpha_0:=\frac{17+\sqrt{33}}{12}.
\]
The function
\[
  \alpha(s):=\frac{2s^2(3-s^2)}{(s-1)^2}
\]
decreases from \(\alpha_0\) to \(1\) on \([s_0,s_+]\).  For
\(1\leq\alpha\leq\alpha_0\), let \(s(\alpha)\) be its unique inverse
on this interval and set
\[
  L_1(\alpha):=
  \frac{s(\alpha)(2-s(\alpha))}{2s(\alpha)-3}.
\]
For \(\alpha_0<\alpha<4\), let \(L_2(\alpha)\) be the unique root
\[
  \sqrt{\alpha}<r<\frac{3\alpha}{\alpha+2}
\]
of
\begin{equation}\label{eq:P-II}
 P_\alpha^{II}(r):=
 \alpha^3(r+1)^2
 -4\alpha^2r(2r^2-r+2)
 +4\alpha r^2(2r^2-r+2)
 -4r^4=0.
\end{equation}
Finally, define
\begin{equation}\label{eq:tight-parallel-ratio}
  \mathcal L(\alpha):=
  \begin{cases}
    L_1(\alpha),&1\leq\alpha\leq\alpha_0,\\
    L_2(\alpha),&\alpha_0<\alpha<4,\\
    \sqrt{\alpha},&\alpha\geq4.
  \end{cases}
\end{equation}
Appendix~\ref{app:branch-derivation} derives the displayed formulas from
the equality conditions of the matching upper bound. The parametrization
for \(L_1\) and the interval specified for \(L_2\) identify the intended
roots; the polynomial equations alone are not root-selection rules.

\subsection{Three-class instances for the lower bound}
\label{sec:lower-construction}

We first make the lower-bound examples in
Figure~\ref{fig:parallel-lower-cases} precise. The construction below
gives finite admissible instances and their limiting costs;
Appendix~\ref{sec:tight-lower} selects the extremal parameters for
each algebraic branch and records their matching certificate contacts.

\paragraph*{What the three classes contribute.}
The construction in the next lemma uses a fraction \(0<y<1\) of the
machines for large jobs \(L\), leaving a fraction \(u=1-y\) without one.
Each large job has duration one and weight \(\alpha\). The medium jobs \(M\) have total
processing load \(mx\) and weight-to-true-time ratio \(\alpha\); the
short jobs \(S\) have total load \(mz\) and ratio \(1\). Medium and large jobs
are unstretched, \(p=\mu\), which we call type \(A\); short jobs are
fully stretched, \(p=\alpha\mu\), which we call type \(B\). The two
types coincide when \(\alpha=1\). All three classes have the same nominal
ratio \(w/\mu=\alpha\), so the list may put short jobs first, medium
jobs second, and large jobs last. The short and medium jobs are divided
into increasingly small jobs; the table gives their limiting costs.

Here \(x\) and \(z\) are loads divided by \(m\), and \(y\) is both a
machine fraction and the large-job load divided by \(m\). In the list
schedule every machine first processes short load \(z\) and medium load
\(x\); each large job therefore completes at \(x+z+1\). In Case~I,
\(x+z\leq u\), medium and short jobs run, in that order, on the
\(um\) machines without a large job. Their respective loads there are
\(x/u\) and \(z/u\); the large jobs complete at time one on their own
machines.

\begin{center}
\small
\renewcommand{\arraystretch}{1.35}
\begin{tabular}{c@{\quad}c@{\quad}c@{\quad}c@{\quad}c}
\hline
\shortstack{class\\list order} &
\shortstack{total load\\divided by \(m\)} &
\(w/p\) &
\shortstack{list cost\\divided by \(m\)} &
\shortstack{Case I comparison cost\\divided by \(m\)}\\
\hline
\(S\) (first) & \(z\) & \(1\) &
\(z^2/2\) & \(xz/u+z^2/(2u)\)\\
\(M\) (second) & \(x\) & \(\alpha\) &
\(\alpha xz+\alpha x^2/2\) & \(\alpha x^2/(2u)\)\\
\(L\) (last) & \(y\) & \(\alpha\) &
\(\alpha y(x+z+1)\) & \(\alpha y\)\\
\hline
\end{tabular}
\end{center}

Each cost entry is a limiting weighted completion cost divided by \(m\).
The squared terms account for work within a block of increasingly small
jobs. The cross terms account for waiting behind another block: multiply
the delayed block's total weight divided by \(m\) by its waiting time.
Thus \(\alpha xz\) is the medium block's waiting-cost contribution behind
short load \(z\), and \(xz/u\) is the short block's waiting-cost contribution
behind medium load \(x/u\) in the comparison. Summing the cost columns gives
\(A_\alpha\) and \(C_\alpha^I\), the limiting list and comparison costs
per machine. For each finite
instance, the corresponding comparison schedule bounds the clairvoyant
optimum from above.

In Case~II, \(x\leq u<x+z\), the large and medium contributions stay the
same. Short jobs first fill the non-large machines from load \(x/u\) to
one, using total short load \(m(u-x)\). Its cost per machine is the sum
of the waiting contribution \(x(u-x)/u\) and within-block contribution
\((u-x)^2/(2u)\). The remaining total short load is \(mr\), where
\(r=x+z-u\). Spreading it evenly over all machines after time one gives
additional load \(r\) on each and cost per machine \(r+r^2/2\).
Combining these contributions gives the limiting comparison cost per
machine:
\[
  C_\alpha^{II}(x,y,z)
  =\alpha y+\frac{\alpha x^2}{2u}
   +\frac{x(u-x)}u+\frac{(u-x)^2}{2u}
   +r+\frac{r^2}{2}.
\]
This is the expanded form of the expression in
Lemma~\ref{lem:three-class}; its subtraction term is an algebraic
correction after substituting \(r=x+z-u\) and \(u=1-y\).
The comparison schedule only needs to be feasible, because its cost
bounds the clairvoyant optimum from above.
Figure~\ref{fig:parallel-lower-cases} shows both schedules, and the proof
below gives the finite instances and the limiting calculation.

\begin{lemma}[Three-class construction]\label{lem:three-class}
Fix \(\alpha\geq1\), \(0<y<1\), \(x,z>0\), and put \(u:=1-y\).
Assume \(x\leq u\), and define
\begin{align*}
 A_\alpha(x,y,z)
 &:=
 \frac{z^2}{2}+\alpha xz+\frac{\alpha x^2}{2}
 +\alpha y(x+z+1),\\
 C_\alpha^{I}(x,y,z)
 &:=
 \alpha y+\frac{\alpha x^2}{2u}
 +\frac{xz}{u}+\frac{z^2}{2u},\\
 C_\alpha^{II}(x,y,z)
 &:=
 \alpha y+\frac{\alpha x^2}{2u}
 +(x+y)z+\frac{z^2}{2}
 -\frac{y(u-x)^2}{2u},
\end{align*}
and
\[
 C_\alpha(x,y,z):=
 \begin{cases}
   C_\alpha^{I}(x,y,z),&x+z\leq u,\\
   C_\alpha^{II}(x,y,z),&x\leq u<x+z.
 \end{cases}
\]
Then
\[
  \mathcal R(\alpha)
  \geq
  \frac{A_\alpha(x,y,z)}{C_\alpha(x,y,z)}.
\]
\end{lemma}

\begin{proof}
First fix rational \(x,y,z\) satisfying the stated inequalities and
choose \(m\) so that \(ym\) is integral. Choose \(N\) through an unbounded
sequence for which \(N/u\) is integral. In Case~II, also require
\(N(u-x)/(uz)\) and \(N(x+z-u)/z\) to be nonnegative integers. Taking
multiples of a common denominator meets these conditions; \(m\) remains
fixed while \(N\) grows. Construct an
instance \(I_{m,N}\) on \(m\) machines with
\(ym\) jobs of type \(L\), \(mN\) jobs of type \(M\), and \(mN\) jobs
of type \(S\), with data
\[
\begin{array}{c|ccc}
 &\mu&p&w\\ \hline
 L&1&1&\alpha\\
 M&x/N&x/N&\alpha x/N\\
 S&z/(\alpha N)&z/N&z/N.
\end{array}
\]
Every job has nominal Smith ratio \(w/\mu=\alpha\).  Hence the maximum
defining \(V_L(I_{m,N})\) is at least the cost of an order in which every
type-\(S\) job precedes every type-\(M\) job, and every type-\(M\) job
precedes every type-\(L\) job.  List scheduling in this order
places \(N\) jobs of type \(S\) and then \(N\) jobs of type \(M\) on
each machine.  The \(L\)-jobs then occupy distinct machines.  Dividing
the resulting cost by \(m\) and sending \(N\) to infinity gives
\[
  \liminf_{N\to\infty}\frac{V_L(I_{m,N})}{m}
  \geq A_\alpha(x,y,z).
\]
The list schedule is illustrated in panels (a) and (c) of
Figure~\ref{fig:parallel-lower-cases}.

For a comparison schedule, put one \(L\)-job alone on each of \(ym\)
machines and distribute the \(M\)-jobs evenly over the other \(um\)
machines.  Their load on each such machine is \(x/u\).

If \(x+z\leq u\), distribute all \(S\)-jobs evenly over the same \(um\)
machines, after the \(M\)-jobs.  The limiting cost per machine is
\[
  \alpha y+\frac{\alpha x^2}{2u}
  +\frac{xz}{u}+\frac{z^2}{2u}
  =C_\alpha^{I}(x,y,z).
\]
Panel (b) of Figure~\ref{fig:parallel-lower-cases} shows this comparison
schedule.

If \(x\leq u<x+z\), first use \(S\)-jobs to raise every non-\(L\)
machine from load \(x/u\) to load \(1\), and distribute the remaining
\(S\)-jobs evenly over all \(m\) machines.  Put \(r:=x+z-u\) for the
remaining total \(S\)-load divided by \(m\).  The \(L\)-jobs contribute
\(\alpha y\), and the \(M\)-jobs contribute \(\alpha x^2/(2u)\).  On the
\(u\)-fraction of non-\(L\) machines, the first \(S\)-portion has total
load \(u-x\); its cross term with \(M\) and its within-class term are
\[
  \frac{x(u-x)}u
  \qquad\text{and}\qquad
  \frac{(u-x)^2}{2u},
\]
respectively.  Every machine now has load \(1\), so the residual
\(S\)-jobs contribute \(r+r^2/2\).  Consequently the limiting comparison
cost per machine is
\begin{align*}
 &\alpha y+\frac{\alpha x^2}{2u}
  +\frac{x(u-x)}u+\frac{(u-x)^2}{2u}+r+\frac{r^2}{2}\\
 &\qquad=
  \alpha y+\frac{\alpha x^2}{2u}
  +(x+y)z+\frac{z^2}{2}
  -\frac{y(u-x)^2}{2u}
  =C_\alpha^{II}(x,y,z).
\end{align*}
Panel (d) of Figure~\ref{fig:parallel-lower-cases} shows the
fill-and-spread operation.
For these divisible values of \(N\), the normalized costs of the chosen
list and comparison schedules equal their displayed limiting expressions plus
\((\alpha x^2+z^2)/(2N)\). With the rational parameters and \(m\) fixed,
send \(N\) to infinity. Since the optimum is no larger than the comparison
cost, this proves the bound for rational parameters. Next approximate
admissible real parameters by rational ones in the same regime. At
\(x+z=u\), the two comparison expressions agree, so continuity also
preserves the bound at the regime boundary. Finally, for any desired
accuracy, select a finite instance and apply Lemma~\ref{lem:remove-ties}
with a perturbation small enough to preserve the chosen list-to-optimum
ratio within that accuracy. This perturbation is chosen after \(m,N\) and may depend on
the selected finite instance.
\end{proof}

\subsection{The core upper-bound certificate}
\label{sec:core-certificate}

\paragraph*{From the blueprint to the certificate.}
We now turn the load comparison from
Section~\ref{sec:intro-techniques} into an upper bound for every instance.
By Lemmas~\ref{lem:equal-ratio-lifting}
and~\ref{lem:endpoint-reduction}, it is enough to consider endpoint
instances normalized by \(w_i=\mu_i\). Recall that type \(A\) means
unstretched, \(p_i=\mu_i\), and type \(B\) means fully stretched,
\(p_i=\alpha\mu_i\). When \(\alpha>1\), a true Smith order places every
type-\(A\) job before every type-\(B\) job.

\paragraph*{Normalization and the prefix--suffix cut.}
For the following estimates, assume \(n>m\) and \(\alpha>1\). Choose a
priority list attaining \(V_L(I)\), and let \(S\) be its list schedule.
Let \(\ell>0\) be the minimum final machine load in
\(S\).  Divide \(\bmu\) and \(\bp\) by \(\ell\), and multiply \(\bw\) by
\(\alpha/\ell\).  This scales every schedule cost by the same factor and
therefore preserves the approximation ratio.  With
\(\kappa_i:=w_i/p_i\), the endpoint property gives
\[
  \min_h \operatorname{load}_h(S)=1,
  \qquad
  \kappa_i=
  \begin{cases}
    \alpha,&i\in A:=\{i:p_i=\mu_i\},\\
    1,&i\in B:=\{i:p_i=\alpha\mu_i\}.
  \end{cases}
\]
For any schedule \(T\) of a job set \(\mathcal J\), and any
\(J\subseteq\mathcal J\), define
\[
  \operatorname{cost}(J;T):=\sum_{i\in J}\kappa_i p_iC_i^T,
  \qquad
  \operatorname{cost}(T):=
  \sum_{i\in\mathcal J}\kappa_i p_iC_i^T.
\]

Fix \(\eta>0\).  On each machine, cut its scheduled job sequence at the
point where its cumulative processing load first reaches \(1\), splitting
a job if the cut passes through it.  Call the part before the cut the
prefix and the remainder the suffix, and subdivide every prefix portion
into consecutive pieces of size at most \(\eta\).  Scale the nominal time
and weight of each piece proportionally, so every piece retains the type
of its original job.  If a cut falls exactly on a job boundary, discard
the resulting zero-length piece.

Let \(I_\eta\) be the resulting fragmented instance on the same \(m\)
machines.  Let \(S_\eta\) be the fragmented schedule obtained from \(S\)
by replacing every split job with consecutive pieces on its original
machine, preserving the original machine sequence; in particular, list
scheduling is not rerun on \(I_\eta\).
When discussing \(I_\eta\) below, we reuse \(A\) and \(B\) for its sets
of type-\(A\) and type-\(B\) pieces, and \(p_i\) denotes the processing
time of piece \(i\).

Every machine of \(S_\eta\) has a prefix of load \(1\) and at most one
suffix piece, which starts at time \(1\).  This is the only
list-scheduling property needed below, and it is inherited from the
original schedule \(S\).  Indeed, if a second original job started after
time \(1\) on some machine of \(S\), then, when that job was assigned, a
machine whose final load is \(1\) had accumulated load at most \(1\).
The chosen machine therefore could not have been least loaded.
Fragmenting in place creates no additional suffix job, and any job
crossed by the cut contributes only its single unsplit suffix portion.

\paragraph*{Upper-bounding the list cost.}
Set \(\beta:=\alpha-1\), the difference between the two normalized
weight-to-time ratios. For a unit prefix with type-\(A\) load \(x\),
the largest mixed-pair contribution occurs when type \(B\) precedes
type \(A\). Define its limiting cost by
\[
  G_\alpha(x):=
  \frac{\alpha x^2}{2}+\frac{(1-x)^2}{2}+\alpha x(1-x)
  =\frac12+\beta x-\frac{\beta x^2}{2}.
\]
The three terms account for pairs within type \(A\), pairs within type
\(B\), and mixed pairs, respectively. Finite prefix pieces also have
nonzero diagonal costs, which we bound before taking the limit.

For machine \(h\), let \(P_h\) be its set of prefix pieces, and let \(x_h\)
be their total \(A\)-load.  Their \(B\)-load is \(1-x_h\).  Bounding every
mixed pair by coefficient \(\alpha\), and using
\(p_i^2\leq\eta p_i\) for prefix pieces, gives
\[
  \operatorname{cost}(P_h;S_\eta)
  \leq
  \frac{\alpha x_h^2}{2}
  +\frac{(1-x_h)^2}{2}
  +\alpha x_h(1-x_h)
  +\frac{\alpha\eta}{2}
  =
  G_\alpha(x_h)+\frac{\alpha\eta}{2}.
\]
Since \(G_\alpha\) is concave, with
\(x:=m^{-1}\sum_hx_h\),
\begin{align}
 \operatorname{cost}(S_\eta)
 \leq{}&
 mG_\alpha(x)+\frac{\alpha\eta m}{2}\notag\\
 &+\sum_{i\in A\cap P^c}\alpha p_i(1+p_i)
 +\sum_{i\in B\cap P^c}p_i(1+p_i),
\label{eq:list-certificate-upper}
\end{align}
where \(P:=\bigcup_hP_h\), and \(P^c\) denotes its complement in the set
of all pieces of \(I_\eta\).

Here \(x\in[0,1]\) is the average type-\(A\) load in the normalized
unit prefixes. In the lower example it is the normalized medium-load
fraction; Section~\ref{sec:lower-construction} uses \(x\) for the unscaled
medium load.

\paragraph*{Lower-bounding the optimum.}
In an optimal schedule for \(I_\eta\), every type-\(A\) piece precedes
every type-\(B\) piece on the same machine: exchanging an adjacent
type-\(B\) piece of size \(q\) followed by a type-\(A\) piece of size
\(p\) decreases the cost by
\(\beta pq\).  If \(L_h\) and \(U_h\) are respectively the total load
and the \(A\)-load on machine \(h\), then
\begin{equation}\label{eq:opt-quadratic}
 \opt(I_\eta)
 =
 \frac12\sum_h(L_h^2+\beta U_h^2)
 +\frac{\alpha}{2}\sum_{i\in A}p_i^2
 +\frac12\sum_{i\in B}p_i^2.
\end{equation}

We call the optimal-cost lower-bound contributions \emph{charges}; they
need not equal the actual completion costs of the pieces to which they
are assigned. To separate prefix and suffix contributions, consider a
squared machine load.
Write \(r\) for its prefix load and \(p\) for a suffix load, both in the
normalized time units. A reference load \(t\) fixes the affine prefix
contribution \(2tr-t^2\). For \(t,p\geq0\), let
\[
  \varphi_t(p):=
  \min_{r\geq0}\{(r+p)^2-2tr+t^2\}
  =
  \begin{cases}
    2tp,&0\leq p\leq t,\\
    p^2+t^2,&p\geq t.
  \end{cases}
\]
The minimum is attained at \(r=\max\{t-p,0\}\): choose \(r=t-p\)
when \(p\leq t\), and \(r=0\) when \(p\geq t\).
For \(p=0\), the affine function \(2tr-t^2\) is the tangent lower bound
for \(r^2\) at \(r=t\).  The function \(\varphi_t(p)\) is the largest
additional charge depending only on a suffix load \(p\) for which
\begin{equation}\label{eq:envelope}
  (r+p)^2\geq2tr-t^2+\varphi_t(p)
  \qquad(r,p\geq0).
\end{equation}
Thus the right-hand side separates a prefix charge, affine in \(r\), from
a suffix-only charge. We call \(\varphi_t\) a \emph{suffix envelope}:
for the chosen prefix contribution, it is the largest suffix-only
contribution valid for every prefix load \(r\). The optimum may place
several suffix pieces on the same machine; we handle them using
superadditivity below.

Fix a candidate ratio \(\rho>1\) and reference loads
\(0\leq\tau\leq\lambda\). The parameter
\(\lambda\) will be applied to the total-load quadratic in the optimal
cost, and \(\tau\) to the additional type-\(A\)-load quadratic. They are
reference processing loads: on the lower example they equal the loads
described in Table~\ref{tab:lower-certificate}, while for an arbitrary
instance they are fixed parameters and need not be actual machine loads.
They depend only on \(\alpha\) and \(\rho\), not on the instance or on
the prefix mixture \(x\) or suffix size \(p\).

Classify the pieces as \(P\) or \(P^c\) according to their positions in
the fragmented schedule \(S_\eta\), irrespective of where they occur in the
optimal schedule.  On machine \(h\) of the optimal schedule, let \(r_h\)
and \(q_h\) be the total loads of pieces in \(P\) and \(P^c\),
respectively, and let \(a_h\) and \(b_h\) be their corresponding
\(A\)-loads.  Because \(P\) has total load \(m\) and total \(A\)-load
\(mx\),
\[
  \sum_hr_h=m,\qquad\sum_ha_h=mx.
\]
The function \(\varphi_t\) is convex and vanishes at zero, hence it is
superadditive on \(\mathbb R_{\geq0}\).  Indeed, for \(x,y>0\), convexity
gives
\[
 \frac{\varphi_t(x)}x,\,
 \frac{\varphi_t(y)}y
 \leq
 \frac{\varphi_t(x+y)}{x+y},
\]
and multiplying by \(x,y\) and adding proves
\(\varphi_t(x)+\varphi_t(y)\leq\varphi_t(x+y)\); the cases \(x=0\) or
\(y=0\) are immediate.  Apply \eqref{eq:envelope} with \(t=\lambda\) to
\(L_h=r_h+q_h\), and with \(t=\tau\) to
\(U_h=a_h+b_h\).  Using this superadditivity in \eqref{eq:opt-quadratic}, retaining the
suffix diagonal terms and dropping only the nonnegative prefix diagonal
terms, yields
\begin{align}
 \opt(I_\eta)\geq{}&
 m\left(
   \lambda-\frac{\lambda^2}{2}
   +\beta\tau x-\frac{\beta\tau^2}{2}
 \right)\notag\\
 &+\sum_{i\in A\cap P^c}
 \left(
   \frac{\alpha p_i^2}{2}
   +\frac{\varphi_\lambda(p_i)}{2}
   +\frac{\beta\varphi_\tau(p_i)}{2}
 \right)\notag\\
 &+\sum_{i\in B\cap P^c}
 \left(
   \frac{p_i^2}{2}
   +\frac{\varphi_\lambda(p_i)}{2}
 \right).
\label{eq:opt-certificate-lower}
\end{align}

Reference loads \(0\leq\tau\leq\lambda\) certify ratio \(\rho>1\) if
\begin{equation}\label{eq:prefix-certificate}
  G_\alpha(x)
  \leq
  \rho\left(
    \lambda-\frac{\lambda^2}{2}
    +\beta\tau x-\frac{\beta\tau^2}{2}
  \right)
  \qquad(0\leq x\leq1)
\end{equation}
and
\begin{equation}\label{eq:suffix-certificate}
  \alpha p(1+p)
  \leq
  \rho\left(
    \frac{\alpha p^2}{2}
    +
    \frac{\varphi_\lambda(p)}{2}
    +\frac{\beta\varphi_\tau(p)}{2}
  \right)
  \qquad(p\geq0).
\end{equation}
The first inequality covers the aggregate prefix, uniformly over \(x\);
the second covers each type-\(A\) suffix, uniformly over its size \(p\).

\begin{lemma}[Prefix--suffix certificate]\label{lem:prefix-suffix}
Fix \(\alpha>1\) and \(\rho>1\).  If a prefix--suffix certificate exists
for \(\alpha\) and \(\rho\), then, for every \(m,n\in\mathbb N\), every endpoint instance
\(I=(\bmu,\bp,\bw;m)\) satisfies
\[
  V_L(I)\leq\rho\opt(I).
\]
\end{lemma}

\paragraph*{Two inequalities suffice.}
Compare \eqref{eq:list-certificate-upper} with
\eqref{eq:opt-certificate-lower}. The prefix contribution and each suffix
contribution are covered separately by the two certificate conditions
stated above. We now apply them.
The first condition charges the entire list-schedule prefix against its
aggregate optimal-schedule charge, uniformly over the possible type-\(A\)
fraction \(x\).  The second charges an individual type-\(A\) suffix piece,
uniformly over its size \(p\).
This second condition also covers type-\(B\) suffix pieces. For fixed
\(p\geq0\), the piecewise formula shows that \(\varphi_t(p)\) is
nondecreasing in \(t\). Hence \(\tau\leq\lambda\) gives
\(\varphi_\tau(p)\leq\varphi_\lambda(p)\), and therefore
\(\varphi_\lambda(p)+\beta\varphi_\tau(p)
\leq\alpha\varphi_\lambda(p)\), since \(\beta=\alpha-1\).
Dividing \eqref{eq:suffix-certificate} by \(\alpha\) thus implies
\[
  p(1+p)
  \leq
  \rho\left(
    \frac{p^2}{2}+\frac{\varphi_\lambda(p)}{2}
  \right).
\]
This compares the list cost of a type-\(B\) suffix piece with its
optimal-cost charge in \eqref{eq:opt-certificate-lower}.
Figure~\ref{fig:prefix-suffix-certificate}
in the overview summarizes the comparison. The equality cases in
Table~\ref{tab:lower-certificate} guide the choice of parameters, and the
two certificate lemmas in Appendix~\ref{sec:tight-upper} check the
inequalities for every load.

\begin{proof}
If \(n\leq m\), every job starts at zero and the claim is immediate.
Otherwise use the normalization and in-place fragmentation described
above. We first compare the fragmented cost ratio with the original one.
For an original job \(i\) of processing time \(p_i=\sum_r p_{ir}\),
splitting it into consecutive pieces decreases its cost by
\[
  \Delta_i:=
  \frac{\kappa_i}{2}
  \left(p_i^2-\sum_rp_{ir}^2\right).
\]
Put \(\Delta:=\sum_i\Delta_i\).  The consecutive, in-place replacement gives
the first identity below.  Replacing every job by the same consecutive
pieces in an optimal schedule for \(I\) gives the second:
\[
  \operatorname{cost}(S_\eta)=V_L(I)-\Delta,
  \qquad
  \opt(I_\eta)\leq\opt(I)-\Delta.
\]
Since \(V_L(I)\geq\opt(I)\),
\begin{equation}\label{eq:fragment-comparison}
  \frac{V_L(I)}{\opt(I)}
  \leq
  \frac{\operatorname{cost}(S_\eta)}{\opt(I_\eta)}.
\end{equation}

\smallskip
\noindent\emph{Combining the bounds.}
Combining \eqref{eq:prefix-certificate},
\eqref{eq:suffix-certificate},
\eqref{eq:list-certificate-upper}, and
\eqref{eq:opt-certificate-lower} gives
\[
  \operatorname{cost}(S_\eta)
  \leq\rho\opt(I_\eta)+\frac{\alpha\eta m}{2}.
\]
Finally, \(\opt(I_\eta)\geq m/2\).  Indeed, the prefix pieces have total
load \(m\); if \(r_h\) is their load on optimal machine \(h\), then
\(\sum_h r_h=m\), so Cauchy--Schwarz and
\eqref{eq:opt-quadratic} give
\(\opt(I_\eta)\geq\frac12\sum_hL_h^2
\geq\frac12\sum_hr_h^2\geq m/2\).  Hence
\[
  \frac{V_L(I)}{\opt(I)}
  \leq
  \frac{\operatorname{cost}(S_\eta)}{\opt(I_\eta)}
  \leq\rho+\alpha\eta.
\]
Letting \(\eta\downarrow0\) proves the lemma.
\end{proof}

\paragraph*{Why two branches appear.}
For the extremal parameter choices in Theorem~\ref{thm:tight-lower},
the two examples in Figure~\ref{fig:parallel-lower-cases} identify where
these inequalities are tight. With the figure's loads, put
\(\ell=x+z\), \(x_0=x/\ell\), \(p_L=1/\ell\), and
\(\tau=x/(u\ell)\). Both examples give equality at prefix mixture
\(x_0\) and large-suffix size \(p_L\).
Case~I has \(\lambda=1/u\) and \(p_L\geq\lambda\), so the suffix
equality lies in the high region. Case~II has
\(\lambda=(\ell+1-u)/\ell\) and \(\tau<p_L<\lambda\), so it lies in the
middle region. The transition is precisely
\(\ell=u\), equivalently \(p_L=\lambda\).
The lower schedules thus supply the reference loads.
Appendix~\ref{sec:tight-lower} gives the full correspondence;
Appendix~\ref{sec:tight-upper} verifies both inequalities for every
mixture and suffix size at \(\rho=\mathcal L(\alpha)\).

\section{Beyond list scheduling: near-optimal robustness}
\label{sec:near-optimal-robustness}

Choosing the machine assignment from the nominal data can improve on the
fixed list rule. The following transfer separates nominal approximation
quality from the loss caused by distortion. Write
\(\operatorname{cost}_{\bq}(S):=\sum_jw_jC_j^S(\bq)\) and
\(\operatorname{OPT}(\bq):=\opt(\bq,\bw,m)\). A fixed schedule specifies
the assignment and order on each machine, and executes them without idle
time, starting each next job when its predecessor completes.

\begin{lemma}[Nominal-to-robust transfer]
\label{lem:nominal-robust-transfer}
Fix \(m\geq1\) and \(\delta\geq0\). If a fixed schedule \(S\) orders
each machine by nonincreasing \(w_j/\mu_j\) and satisfies
\(\operatorname{cost}_{\bmu}(S)\leq(1+\delta)\operatorname{OPT}(\bmu)\),
then, simultaneously for every \(\alpha\geq1\) and
\(\bmu\leq\bp\leq\alpha\bmu\),
\[
 \operatorname{cost}_{\bp}(S)
 \leq(\sqrt{\alpha}+\alpha\delta)\operatorname{OPT}(\bp).
\]
\end{lemma}

\begin{proof}
The claim is immediate for \(\alpha=1\), so assume \(\alpha>1\) and put
\(q:=\sqrt{\alpha}\).  We first prove the claim for an endpoint vector
\[
  p_j^\sigma=\sigma_j\mu_j,
  \qquad
  \sigma\in\Sigma:=\{1,\alpha\}^n.
\]
Let
\[
  r_j:=\frac{w_j}{\mu_j},
  \qquad
  J_\theta:=\{j:r_j\geq\theta\}
  \quad(\theta\geq0).
\]
The layer identity
\[
  w_j=\int_0^\infty
       \mu_j\mathbf 1_{\{j\in J_\theta\}}\,d\theta
\]
gives, for every processing vector \(\bq\) and schedule \(R\),
\begin{equation}\label{eq:robust-transfer-layer}
  \operatorname{cost}_{\bq}(R)
  =
  \int_0^\infty
  \sum_{j\in J_\theta}\mu_jC_j^R(\bq)\,d\theta.
\end{equation}

Fix \(\sigma\), and let
\[
  A:=\{j:\sigma_j=1\},
  \qquad
  B:=\{j:\sigma_j=\alpha\}.
\]
Write \(P=(P_1,\ldots,P_m)\) for the assignment of \(S\), and define
\[
  s_k(\theta):=\sum_{j\in A\cap J_\theta\cap P_k}\mu_j,
  \qquad
  \ell_k(\theta):=\sum_{j\in B\cap J_\theta\cap P_k}\mu_j,
  \qquad
  t_k(\theta):=s_k(\theta)+\ell_k(\theta).
\]
Set
\[
 \begin{aligned}
  D&:=
  \frac12\int_0^\infty
  \left(
    \sum_{j\in A\cap J_\theta}\mu_j^2
    +\alpha\sum_{j\in B\cap J_\theta}\mu_j^2
  \right)d\theta,\\
  X&:=\int_0^\infty\sum_{k=1}^m t_k(\theta)^2\,d\theta,
  \qquad
  X_A:=\int_0^\infty\sum_{k=1}^m s_k(\theta)^2\,d\theta.
 \end{aligned}
\]
Because \(S\) is nominally Smith-ordered on each machine,
\(J_\theta\cap P_k\) is a prefix of that machine.  Expanding the pairwise
contributions in \eqref{eq:robust-transfer-layer}, and charging every
mixed \(A\)-\(B\) pair by at most
\(\alpha\mu_a\mu_b\), yields
\begin{equation}\label{eq:robust-transfer-alg}
  \operatorname{cost}_{\bp^\sigma}(S)-D
  \leq
  \frac12\bigl(\alpha X-(\alpha-1)X_A\bigr).
\end{equation}

We next construct a lower bound for the clairvoyant optimum.  For an
assignment \(Q=(Q_1,\ldots,Q_m)\), let
\[
  a_k^Q(\theta):=\sum_{j\in A\cap J_\theta\cap Q_k}\mu_j,
  \qquad
  b_k^Q(\theta):=\sum_{j\in B\cap J_\theta\cap Q_k}\mu_j
\]
and define
\[
  F_\sigma(Q):=
  \int_0^\infty\sum_{k=1}^m
  \left(
    (a_k^Q)^2+2a_k^Qb_k^Q+\alpha(b_k^Q)^2
  \right)d\theta,
\]
where the argument \(\theta\) is suppressed inside the integrand.
Choose an assignment \(Q^*\) minimizing \(F_\sigma\), and put
\[
 \begin{aligned}
  Y&:=\int_0^\infty\sum_k(a_k^{Q^*}(\theta))^2\,d\theta,\\
  C&:=\int_0^\infty\sum_k
      a_k^{Q^*}(\theta)b_k^{Q^*}(\theta)\,d\theta,\\
  Z&:=\int_0^\infty\sum_k(b_k^{Q^*}(\theta))^2\,d\theta.
 \end{aligned}
\]
For any schedule with assignment \(Q\), delete the jobs outside
\(J_\theta\) and expand the remaining pairwise contributions.  Every
mixed pair contributes at least \(\mu_a\mu_b\).  Integrating and then
minimizing over \(Q\) gives
\begin{equation}\label{eq:robust-transfer-opt}
  \operatorname{OPT}(\bp^\sigma)
  \geq
  D+\frac12F_\sigma(Q^*)
  =
  D+\frac12(Y+2C+\alpha Z).
\end{equation}

We now use the nominal guarantee.  Let
\[
  D_0:=
  \frac12\int_0^\infty
  \sum_{j\in J_\theta}\mu_j^2\,d\theta,
\]
so \(D_0\leq D\).  The prefix property and
\eqref{eq:robust-transfer-layer} imply
\[
  \operatorname{cost}_{\bmu}(S)=D_0+\frac12X.
\]
Assign the jobs according to \(Q^*\) and order each machine by
nonincreasing \(r_j\); call this schedule \(\overline S\).  Then
\[
  \operatorname{cost}_{\bmu}(\overline S)
  =D_0+\frac12(Y+2C+Z).
\]
The assumed nominal approximation therefore gives
\begin{equation}\label{eq:robust-transfer-X}
  X
  \leq
  (1+\delta)(Y+2C+Z)+2\delta D_0.
\end{equation}

It remains to relate the assignment of \(S\) to \(Q^*\).  Form the hybrid
assignment
\[
  Q_k':=(A\cap P_k)\cup(B\cap Q_k^*)
  \qquad(k\in[m]).
\]
The optimality of \(Q^*\) for \(F_\sigma\) gives
\[
  Y+2C
  \leq
  X_A+
  2\int_0^\infty\sum_k
    s_k(\theta)b_k^{Q^*}(\theta)\,d\theta.
\]
Using
\[
  2q s_k(\theta)b_k^{Q^*}(\theta)
  \leq
  s_k(\theta)^2+q^2(b_k^{Q^*}(\theta))^2
\]
and integrating yields
\begin{equation}\label{eq:robust-transfer-key}
  q(Y+2C)
  \leq
  (q+1)X_A+q^2Z.
\end{equation}
Since \(\alpha=q^2\), this is equivalent to
\begin{equation}\label{eq:robust-transfer-key-consequence}
  \alpha(Y+2C+Z)-(\alpha-1)X_A
  \leq
  q(Y+2C+\alpha Z).
\end{equation}

Combining \eqref{eq:robust-transfer-alg},
\eqref{eq:robust-transfer-X}, and
\eqref{eq:robust-transfer-key-consequence}, and using
\(\alpha\geq1\), gives
\[
 \begin{aligned}
  \operatorname{cost}_{\bp^\sigma}(S)-D
  &\leq
  \frac{q+\alpha\delta}{2}(Y+2C+\alpha Z)
  +\alpha\delta D_0.
 \end{aligned}
\]
Together with \eqref{eq:robust-transfer-opt} and \(D_0\leq D\), this
implies
\[
 \begin{aligned}
  \operatorname{cost}_{\bp^\sigma}(S)
  &\leq
  D+(q+\alpha\delta)
       \bigl(\operatorname{OPT}(\bp^\sigma)-D\bigr)
       +\alpha\delta D_0\\
  &=
  (q+\alpha\delta)\operatorname{OPT}(\bp^\sigma)
  -(q-1)D-\alpha\delta(D-D_0)\\
  &\leq
  (q+\alpha\delta)\operatorname{OPT}(\bp^\sigma).
 \end{aligned}
\]

Finally, fix any \(\bmu\leq\bp\leq\alpha\bmu\), and write it as the
convex combination of endpoint vectors
\[
  \bp=\sum_{\sigma\in\Sigma}c_\sigma\bp^\sigma,
  \qquad
  c_\sigma\geq0,
  \qquad
  \sum_\sigma c_\sigma=1.
\]
For example, take
\[
  \lambda_j:=\frac{p_j-\mu_j}{(\alpha-1)\mu_j},
  \qquad
  c_\sigma:=
  \prod_{\{j:\sigma_j=\alpha\}}\lambda_j
  \prod_{\{j:\sigma_j=1\}}(1-\lambda_j).
\]
The cost of the fixed schedule \(S\) is linear in \(\bp\), whereas
\(\operatorname{OPT}(\bp)\), being the minimum of linear schedule-cost
functions, is concave.  Consequently,
\[
 \begin{aligned}
  \operatorname{cost}_{\bp}(S)
  &=\sum_\sigma c_\sigma
    \operatorname{cost}_{\bp^\sigma}(S)\\
  &\leq
  (q+\alpha\delta)
  \sum_\sigma c_\sigma\operatorname{OPT}(\bp^\sigma)\\
  &\leq
  (q+\alpha\delta)\operatorname{OPT}(\bp).
 \end{aligned}
\]
This proves the claim.
\end{proof}

\begin{corollary}[Optimal nominal assignment]
\label{cor:exact-nominal-optimum}
An exact nominal optimum, ordered by nominal Smith ratios on each
machine, is a fixed schedule depending only on \((\bmu,\bw,m)\) whose
realized ratio is at most \(\sqrt{\alpha}\) simultaneously for every
\(\alpha\geq1\). For each fixed \(m\), this is the optimal worst-case
guarantee among deterministic nonpreemptive algorithms.
\end{corollary}

Nominal Smith reordering preserves optimality, so the corollary follows
by setting \(\delta=0\); Lemma~\ref{lem:det-lb} supplies the matching
lower bound. Exact nominal optimization may be expensive.

For a polynomial-time consequence, fix \(A\geq1\) and
\(\varepsilon>0\), and set \(\delta=\varepsilon/\sqrt A\). Apply the
Skutella--Woeginger nominal PTAS~\cite{skutella2000ptas}, then keep its
assignment and reorder each machine by nominal Smith ratios. Reordering
does not increase nominal cost, and the lemma gives
\(\sqrt{\alpha}+\alpha\delta\leq(1+\varepsilon)\sqrt{\alpha}\)
simultaneously for \(1\leq\alpha\leq A\). In the standard binary
rational-input model, this runs in polynomial time for fixed
\(A,\varepsilon\), with \(m\) part of the input. Taking \(A=4\)
proves Theorem~\ref{thm:intro-robust-ptas}; the special role of \(4\) is
that \(\rwspt\) attains the exact benchmark from that point onward.
The PTAS accuracy depends on the chosen range \(A\); achieving the
same factor with one polynomial-time algorithm simultaneously over all
\(\alpha\geq1\) remains open.

\section{Conclusion}

We revisited one of the most classical scheduling rules through the lens of
processing-time distortion.  On one machine, the exact robustness of
\(\rwspt\) is \(\sqrt{\alpha}\), matching the lower bound for every
deterministic nonpreemptive algorithm. On identical parallel machines,
the list rule's exact worst-case ratio, with the supremum over machine
counts, is the algebraic curve \(\mathcal L(\alpha)\): it
starts at the Kawaguchi--Kyan constant when \(\alpha=1\), changes branch at
\(\alpha_0\), and equals \(\sqrt{\alpha}\) for every \(\alpha\geq4\), when
\(\rwspt\) is optimal among deterministic nonpreemptive algorithms for each fixed
machine count.

Common-ratio lifting and endpoint reduction restrict the analysis to jobs
with the same nominal Smith ratio and processing times equal to either
their nominal values or their \(\alpha\) multiples. Cutting
the list schedule at its minimum final machine load then separates a
prefix of the same load on every machine from at most one remaining job
piece per machine. The lower-bound examples identify equality cases for
the two upper-bound inequalities. Verifying these inequalities for every
prefix mixture and suffix size gives a matching bound; the two packing
regimes explain the algebraic branches.

Beyond the fixed list rule, the nominal-to-robust transfer applies to any
fixed, nominally Smith-ordered schedule with nominal approximation ratio
\(1+\delta\), giving realized ratio at most
\(\sqrt{\alpha}+\alpha\delta\). An exact nominal optimum, ordered by nominal
Smith ratios on each machine, sets \(\delta=0\) and gives a fixed schedule,
independent of \(\alpha\), with realized ratio at most \(\sqrt{\alpha}\)
simultaneously for every \(\alpha\geq1\), achieving the optimal
deterministic guarantee for each fixed \(m\). For every fixed
\(\varepsilon>0\), a nominal PTAS computes a fixed assignment and
per-machine orders in deterministic polynomial time. The same schedule,
independent of \(\alpha\), has a \((1+\varepsilon)\sqrt{\alpha}\) guarantee
simultaneously for \(1\leq\alpha\leq4\). For \(\alpha\geq4\), the generally
different \(\rwspt\) list schedule attains the lower benchmark for
every fixed machine count. It remains
open whether, for every fixed \(\varepsilon>0\), one deterministic
polynomial-time algorithm, independent of \(\alpha\) and with \(m\) as
part of the input, can achieve \((1+\varepsilon)\sqrt{\alpha}\)
simultaneously for every \(\alpha\geq1\).

\section*{Declaration of AI Assistance}
Generative AI tools provided by OpenAI were used to assist with literature
searches and drafting portions of the related-work discussion, and with the
presentation and formalization of some mathematical claims.  The authors
reviewed and verified all AI-assisted material and every cited source, and
assume responsibility for the literature review and related-work discussion
and for the correctness of all theorems, algorithms, proofs, and other content.

\clearpage
\bibliographystyle{plainurl}
\bibliography{ref}

\clearpage
\appendix
\section{Properties of the parallel-machine target ratio}
\label{sec:parallel}

We now define the worst-case parallel-machine ratio
\[
  \mathcal R(\alpha):=
  \sup_{m,n\in\mathbb N}
  \ \sup_{\substack{\bmu,\bp,\bw\in\Rpos^n\\
                     \bmu\leq\bp\leq\alpha\bmu}}
  \frac{V_L(\bmu,\bp,\bw,m)}{\opt(\bp,\bw,m)}.
\]
Lemma~\ref{lem:remove-ties} below shows that this supremum is unchanged if
all nominal ratios are required to be distinct.

Section~\ref{sec:exact-curve} defines the branches and their selection
rules. The first branch also satisfies the following polynomial, with
its root selected by the parametrization in that section:

\begin{align}
 P_\alpha^{I}(r):={}&
 (8\alpha-108)r^4+(4\alpha+216)r^3 \notag\\
 &+(\alpha^2-24\alpha-108)r^2
 +(2\alpha^2+4\alpha)r+\alpha^2+8\alpha=0.
\label{eq:P-I}
\end{align}

Appendix~\ref{app:branch-derivation} derives the equality conditions and
eliminations.

\begin{lemma}[Algebraic branches]
\label{lem:algebraic-branches}
The definitions in Section~\ref{sec:exact-curve} are well posed.
The two branches meet continuously:
\[
  \lim_{\alpha\downarrow\alpha_0}L_2(\alpha)=L_1(\alpha_0).
\]
Moreover, \(L_1(1)=(1+\sqrt2)/2\), and
\(\lim_{\alpha\uparrow4}L_2(\alpha)=2\).
For \(\alpha_0<\alpha<4\), writing \(\rho=L_2(\alpha)\), one also has
\begin{equation}\label{eq:case-II-feasible-branch}
  -\alpha^2\rho^2+4\alpha^2\rho-\alpha^2
  -4\alpha\rho^2+2\rho^3>0.
\end{equation}
For later use, denote the left-hand side, with \(\rho\) replaced by \(r\),
by \(F_\alpha(r)\).
\end{lemma}

\begin{proof}
For the first branch,
\[
 \alpha'(s)=
 -\frac{4s(s+1)(s^2-3s+3)}{(s-1)^3}<0
 \qquad(s_0\leq s\leq s_+),
\]
and substitution gives
\(\alpha(s_0)=\alpha_0\), \(\alpha(s_+)=1\), and
\(L_1(1)=(1+\sqrt2)/2\).  Hence the inverse \(s(\alpha)\) is well
defined.  Appendix~\ref{app:branch-derivation} derives
\eqref{eq:P-I} from the active contact equations.

For the second branch, the two endpoint evaluations in
Appendix~\ref{app:L2-root-audit} have opposite signs, so a root exists in
the defining interval.  That appendix uses the multiple-root resultant
and one Sturm count to prove that this root is unique and simple throughout
\(\alpha_0<\alpha<4\).  It also proves the two endpoint limits, continuity
at \(\alpha_0\), and the feasible-branch inequality
\eqref{eq:case-II-feasible-branch}.
\end{proof}

\section{Extremal parameters and the tight lower bound}\label{sec:tight-lower}

Section~\ref{sec:lower-construction} gives the three-class construction
and its limiting costs. Here we first justify removing nominal-ratio
ties, then choose the parameters attaining each algebraic branch and
derive their correspondence with the upper-bound certificate.

\begin{lemma}[Removing ties]\label{lem:remove-ties}
Fix \(m\in\mathbb N\).  Let \(I=(\bmu,\bp,\bw;m)\) have a common nominal
Smith ratio, and let
\(\bpi=(\pi_1,\ldots,\pi_n)\) be any priority order.  There is a family
\((I_\varepsilon)_{\varepsilon>0}\) of instances on the same \(m\)
machines, with the same \(\bmu\) and \(\bp\), with pairwise distinct
nominal Smith ratios, whose
unique \(\rwspt\) order is \(\bpi\), and such that
\[
 \lim_{\varepsilon\downarrow0}
 \frac{V_m(\bp,\bw^\varepsilon,\bpi)}{\opt(I_\varepsilon)}
 =
 \frac{V_m(\bp,\bw,\bpi)}{\opt(I)}.
\]
\end{lemma}

\begin{proof}
Set \(I_\varepsilon:=(\bmu,\bp,\bw^\varepsilon;m)\), where
\[
 w_{\pi_k}^{\varepsilon}
 :=
 w_{\pi_k}\bigl(1+\varepsilon(n-k)\bigr),
 \qquad k=1,\ldots,n.
\]
Because the original ratios are equal,
\[
 \frac{w_{\pi_1}^{\varepsilon}}{\mu_{\pi_1}}
 >
 \cdots
 >
 \frac{w_{\pi_n}^{\varepsilon}}{\mu_{\pi_n}},
\]
so \(\bpi\) is the unique nominal order.  The processing-time vector and
all distortion constraints are unchanged.  The cost of the resulting
list schedule is linear, and hence continuous, in the weights.  The
optimal cost is the minimum of finitely many linear functions of the
weights and is therefore continuous as well.  Taking
\(\varepsilon\downarrow0\) proves the claim.
\end{proof}

\begin{theorem}[Tight lower bound]\label{thm:tight-lower}
For every \(\alpha\geq1\),
\[
  \mathcal R(\alpha)\geq\mathcal L(\alpha).
\]
\end{theorem}

\begin{proof}
First suppose \(1\leq\alpha\leq\alpha_0\), and abbreviate
\(s:=s(\alpha)\).  In Lemma~\ref{lem:three-class}, choose
\[
  u:=\frac{2-s}{2s-3},
  \qquad
  x:=\frac{2(3-s^2)}{(s+1)(2s-3)},
  \qquad
  z:=sx,
  \qquad
  y:=1-u.
\]
Because \(s_0\leq s\leq s_+\), all four parameters are positive and
\[
  x+z\leq u
  \quad\Longleftrightarrow\quad
  2s^2-s-4\geq0.
\]
Thus the construction is in Case~I.  Substituting the displayed values
and
\[
  \alpha=\frac{2s^2(3-s^2)}{(s-1)^2}
\]
into the two quadratic forms gives the transparent factorization
\[
  A_\alpha(x,y,z)
  =
  \frac{2s^2(2-s)(3-s^2)}
       {(s-1)(2s-3)},
  \qquad
  C_\alpha^{I}(x,y,z)
  =
  \frac{2s(3-s^2)}{s-1}.
\]
Their ratio is
\[
  \frac{A_\alpha(x,y,z)}
       {C_\alpha^{I}(x,y,z)}
  =
  \frac{s(2-s)}{2s-3}
  =
  L_1(\alpha).
\]
Lemma~\ref{lem:three-class} proves the first branch.

Now suppose \(\alpha_0<\alpha<4\) and put
\(\rho:=L_2(\alpha)\).  Define
\[
  \xi:=
  \frac{3\alpha-\alpha\rho-2\rho}
       {(2\rho-\alpha)(\rho-1)},
  \qquad
  u:=\frac{\alpha(\rho-1)}{2(\alpha-\rho)},
  \qquad
  \zeta:=\xi\left(1+\frac{\rho}{u}\right),
\]
and set
\[
  x:=u\xi,\qquad z:=u(\zeta-\xi),\qquad y:=1-u.
\]
The defining interval for \(\rho\) implies
\[
 \rho>1,\quad
 \rho^2-\alpha>0,\quad
 \alpha-\rho>0,\quad
 2\rho-\alpha>0,\quad
 3\alpha-(\alpha+2)\rho>0.
\]
Hence \(\xi,u>0\).  The upper bounds on these quantities and the reason
for the additional feasible-branch inequality
\eqref{eq:case-II-feasible-branch} are exposed by
\begin{align}
  1-\xi
  &=
  \frac{2(\rho^2-\alpha)}
       {(2\rho-\alpha)(\rho-1)}>0,\label{eq:L2-xi-factor}\\
  1-u
  &=
  \frac{3\alpha-(\alpha+2)\rho}
       {2(\alpha-\rho)}>0,\label{eq:L2-u-factor}\\
  \zeta-1
  &=
  \frac{2F_\alpha(\rho)}
       {\alpha(2\rho-\alpha)(\rho-1)^2}>0,\label{eq:L2-zeta-factor}
\end{align}
where \(F_\alpha(\rho)\) is the left-hand side of
\eqref{eq:case-II-feasible-branch}.  Thus
\[
  0<\xi<1,\qquad 0<u<1,\qquad \zeta>1.
\]
Consequently \(x<u<x+z\), so this is Case~II.  Finally, direct
factorization gives
\[
  A_\alpha(x,y,z)-\rho C_\alpha^{II}(x,y,z)
  =
  \frac{\bigl(3\alpha-(\alpha+2)\rho\bigr)
        P_\alpha^{II}(\rho)}
       {8(\alpha-\rho)^2(2\rho-\alpha)(\rho-1)}
  =0.
\]
Appendix~\ref{app:L2-lower-audit} records the associated contact
identities and the same factorization in the variables of the normalized
certificate.
This proves the second branch.

For \(\alpha\geq4\), Lemma~\ref{lem:det-lb} gives the lower bound
\(\sqrt{\alpha}\) for every fixed machine count, and hence also for
\(\mathcal R\). These are the remaining values of \(\mathcal L\).
\end{proof}

\paragraph*{Reading the upper-bound parameters from the lower schedules.}
For the two branches with \(1\leq\alpha<4\), let \(\rho\) be the
corresponding value of \(L_1(\alpha)\) or \(L_2(\alpha)\). The lower
schedule identifies where the matching upper-bound inequalities must be
equalities. We first express its loads in the units used by the upper
proof. Put
\[
  \ell:=x+z.
\]
This is the minimum final load in the three-class list schedule: machines
without a large job finish at \(\ell\), and those with one finish at
\(\ell+1\). Divide all nominal times, processing times, and weights by
\(\ell\). This scales every schedule cost by \(1/\ell^2\) and preserves
the weight-to-processing-time ratios \(w_j/p_j\). The work before the
cut then has load one on every list-schedule machine. Its medium-job
portion has normalized load
\[
  x_0:=\frac{x}{\ell},
\]
which, for \(\alpha>1\), is the unstretched (type-\(A\)) fraction of the
unit prefix. Each large job starts at the cut and is an unstretched suffix
of size
\[
  p_L:=\frac1\ell.
\]
The remaining parameters come from the \emph{comparison} schedule.
The reference load \(\tau=x/(u\ell)\) is the normalized medium load on
each machine without a large job. The reference load \(\lambda\) is the
normalized total load on those machines in Case~I, and on every machine
in Case~II. All four quantities \(x_0,p_L,\tau,\lambda\) measure
processing load after division by \(\ell\). Prefix and suffix labels
always refer to the list schedule, even when these pieces change position
in the comparison schedule.

In Case~I, each non-large machine has total load
\(\lambda=\ell/(u\ell)=1/u\) and medium-job load \(\tau\), while each
large-job machine contains only the suffix \(p_L\). In Case~II, the
fill-and-spread operation gives every comparison machine total load
\(\lambda=(\ell+y)/\ell\). A non-large machine still has medium-job
load \(\tau\); a large-job machine has its suffix \(p_L\) followed by
short-job prefix work of load \(\lambda-p_L\).
At \(\alpha=1\), all jobs are unstretched. The same formulas still
describe the medium-job loads, while the type-dependent terms in the
certificate vanish; see Remark~\ref{rem:kawaguchi-kyan}.

Recall that a \emph{contact} is a prefix mixture or suffix size at which one of
the upper-bound inequalities holds with equality. The suffix inequality
uses different quadratic expressions in the size ranges \(p\leq\tau\),
\(\tau\leq p\leq\lambda\), and \(p\geq\lambda\); we call the last two
the middle and high regions. The symbols \(p_0\) and \(p_1\) name the
suffix equality sizes in the middle and high regions, respectively.
Section~\ref{sec:core-certificate} defines these inequalities. The table
summarizes the loads and equality sizes supplied by the two lower examples.

\begin{table}[ht]
\centering
\small
\renewcommand{\arraystretch}{1.3}
\begin{tabular}{c@{\qquad}c@{\qquad}c@{\qquad}c}
\hline
branch & lower regime & \shortstack{normalized\\comparison loads} &
\shortstack{suffix equality\\and size region}\\
\hline
\(L_1\) &
\(\ell\leq u\) &
\(\displaystyle \tau=\frac{x}{u\ell},\
  \lambda=\frac1u\) &
\(\displaystyle p_L=p_1\geq\lambda\) (high)\\
\(L_2\) &
\(u<\ell\) &
\(\displaystyle \tau=\frac{x}{u\ell},\
  \lambda=\frac{\ell+y}{\ell}\) &
\(\displaystyle \tau<p_L=p_0<\lambda\) (middle)\\
\hline
\end{tabular}
\caption{The lower-bound/certificate correspondence. The comparison
schedule supplies the reference loads \(\tau,\lambda\). Its large job,
of normalized size \(p_L\), gives equality in the suffix inequality in
the indicated size region.}
\label{tab:lower-certificate}
\end{table}

With these reference loads, the upper proof's separation of squared loads
is exact on the comparison schedules. The next identities connect
them to the formulas defining the two branches.

For the first branch, substituting the parameters from the proof above
gives
\begin{equation}\label{eq:lower-L1-contacts}
  x_0=\frac1{s+1},\qquad
  \tau=\frac{2s-3}{(2-s)(s+1)},\qquad
  \lambda=\frac{2s-3}{2-s},\qquad
  p_L=\frac1{2(\rho-1)}=:p_1.
\end{equation}
Thus equality occurs at prefix mixture \(x_0\) and suffix size \(p_1\)
in the high region. Moreover,
\[
  \ell\leq u
  \quad\Longleftrightarrow\quad
  p_L\geq\lambda.
\]

For the second branch, write \(\rho=L_2(\alpha)\) and
\[
  D:=3\alpha\rho-\alpha-2\rho^2,
  \qquad
  c:=\rho(2\alpha-1)-2\alpha.
\]
The parameters in the proof above, together with
\(P_\alpha^{II}(\rho)=0\), give
\begin{equation}\label{eq:lower-L2-contacts}
  x_0=\frac{\alpha(\rho-1)}D,\qquad
  \tau=\frac{2(\alpha-\rho)}D,\qquad
  \lambda=\frac{2\rho(\alpha-1)}D,\qquad
  p_L=\frac{\alpha-\rho\lambda}{c}=:p_0.
\end{equation}
The exact reductions to these four expressions are listed in
Appendix~\ref{app:L2-lower-audit}.
Now the lower-witness condition \(x<u<\ell\) becomes
\[
  \tau<p_L<\lambda,
\]
so equality for the large suffix occurs at size \(p_0\) in the middle
region.
Consequently the change from \(L_1\) to \(L_2\) occurs precisely at
\[
  p_L=\lambda
  \quad\Longleftrightarrow\quad
  \ell=u.
\]
This correspondence is the organizing principle of the upper bound:
rather than guessing \(\tau\) and \(\lambda\), we read them from the
extremal lower schedule and prove that the resulting inequalities hold
for every instance.

\section{Certificate parameters and the tight upper bound}
\label{sec:tight-upper}

Section~\ref{sec:core-certificate} proves the general scheduling
certificate. We now give its classical specialization and verify the
scalar inequalities for the parameters associated with each branch.

\begin{remark}[A certificate-based derivation of the Kawaguchi--Kyan bound]
\label{rem:kawaguchi-kyan}
The proof of Lemma~\ref{lem:prefix-suffix} also applies at
\(\alpha=1\) under its natural one-type specialization: then
\(\beta=0\), all cost coefficients equal one, and the parameter
\(\tau\) disappears.  Put
\[
  \rho:=\rho_{\mathrm{KK}},
  \qquad
  \lambda:=\sqrt2.
\]
The prefix condition holds with equality, since
\[
  \rho\left(\lambda-\frac{\lambda^2}{2}\right)
  =\rho(\sqrt2-1)=\frac12.
\]
For \(0\leq p\leq\lambda\), the suffix residual satisfies
\[
 \rho\left(\frac{p^2}{2}+\lambda p\right)-p(1+p)
 =
 p\left[
   \rho\sqrt2-1-\left(1-\frac{\rho}{2}\right)p
 \right]
 \geq
 \frac{2-\sqrt2}{4}\,p
 \geq0,
\]
where the bracket is decreasing in \(p\) and the displayed lower bound
is its value at \(p=\sqrt2\).  For \(p\geq\lambda\),
\[
 \rho\left(p^2+\frac{\lambda^2}{2}\right)-p(1+p)
 =
 (\rho-1)(p-1-\sqrt2)^2
 \geq0.
\]
Thus the certificate proves the ratio \(\rho\) for common-ratio
instances, and Lemma~\ref{lem:equal-ratio-lifting} extends it
to arbitrary WSPT instances.  Together with the \(\alpha=1\) case of
Theorem~\ref{thm:tight-lower}, this gives a certificate-based derivation
of the tight classical bound of Kawaguchi and
Kyan~\cite{kawaguchi1986worst}.  Whereas Schwiegelshohn's simplified proof
reduces to a canonical worst-case instance
\cite{schwiegelshohn2011alternative}, our reusable prefix--suffix lemma
leaves only the two scalar checks above.  More importantly, the same
certificate proves the distorted bound after introducing a second job type
and a second envelope parameter.
\end{remark}

\paragraph*{Checking the certificate.}
The certificate has now reduced the scheduling argument to two scalar
inequalities. The prefix residual can be written as a square. Set
\(x_0:=1-\rho\tau\) and
\[
 E_{\rm pre}:=
 \rho\bigl(2\lambda-\lambda^2+2\beta\tau
          -\beta(\rho+1)\tau^2\bigr)-\alpha.
\]
Completing the square gives
\begin{equation}\label{eq:prefix-residual-square}
 \rho\left(\lambda-\frac{\lambda^2}{2}
          +\beta\tau x-\frac{\beta\tau^2}{2}\right)-G_\alpha(x)
 =\frac{\beta}{2}(x-x_0)^2+\frac{E_{\rm pre}}{2}.
\end{equation}
For both branches below \(\alpha=4\), the parameters in the following
lemmas satisfy \(E_{\rm pre}=0\), as verified in
Appendices~\ref{app:L2-certificate-audit} and~\ref{app:L1-certificate-audit}.
Since \(\beta>0\), the prefix inequality then holds for every
\(x\in[0,1]\), with equality at \(x_0\).

For the suffix condition, suppose \(0\leq\tau\leq\lambda\) and define
\[
 \Psi_{\rho,\tau,\lambda}(p):=
 \rho\left(
   \frac{\alpha p^2}{2}
   +\frac{\varphi_\lambda(p)}{2}
   +\frac{\beta\varphi_\tau(p)}{2}
 \right)-\alpha p(1+p).
\]
The envelopes change formula at \(\tau\) and \(\lambda\), giving
\begin{equation}\label{eq:suffix-residual-sosa}
 \Psi_{\rho,\tau,\lambda}(p)=
\begin{cases}
 p\!\left[
   \alpha\left(\dfrac{\rho}{2}-1\right)p
   +\rho(\lambda+\beta\tau)-\alpha
 \right],
 &0\leq p\leq\tau,\\[2mm]
 \dfrac{c}{2}p^2+(\rho\lambda-\alpha)p
 +\dfrac{\rho\beta\tau^2}{2},
 &\tau\leq p\leq\lambda,\\[2mm]
 \alpha(\rho-1)p^2-\alpha p
 +\dfrac{\rho}{2}(\lambda^2+\beta\tau^2),
 &p\geq\lambda,
 \end{cases}
\end{equation}
where \(c:=\rho(2\alpha-1)-2\alpha\). Thus the suffix condition is
exactly \(\Psi_{\rho,\tau,\lambda}(p)\geq0\) for all \(p\geq0\).
These three regions classify the size of a suffix piece; they are
separate from the short, medium, and large job classes in the examples
used to prove the lower bound in Section~\ref{sec:lower-construction}.
The two lemmas below verify the inequalities for every prefix mixture
and every suffix size, including away from the contacts in
Table~\ref{tab:lower-certificate}. The \(L_2\) lower witness makes the
middle-region suffix residual tangent to zero at its normalized large-job
size \(p_0\); the \(L_1\) witness makes the high-region residual tangent
at \(p_1\).
In both cases the prefix residual vanishes at the lower witness's
type-\(A\) fraction \(x_0\).

\begin{lemma}[The \(L_2\) certificate]\label{lem:L2-certificate}
Let \(\alpha_0<\alpha<4\) and \(\rho:=L_2(\alpha)\).  Put
\[
 D:=3\alpha\rho-\alpha-2\rho^2,
 \qquad
 \tau:=\frac{2(\alpha-\rho)}{D},
 \qquad
 \lambda:=\frac{2\rho(\alpha-1)}{D}.
\]
Then \(0<\tau<\lambda\), and
\eqref{eq:prefix-certificate}--\eqref{eq:suffix-certificate} hold.
\end{lemma}

\begin{proof}
Let \(c:=\rho(2\alpha-1)-2\alpha\) and
\[
 p_0:=\frac{\alpha-\rho\lambda}{c}.
\]
By \eqref{eq:lower-L2-contacts}, the displayed \(\tau,\lambda\) are the
two loads read from the normalized Case~II comparison schedule, and
\(p_0=p_L\) is its normalized large-job size.  It remains to show that
the two certificate inequalities, which are tight at these lower-witness
contacts, remain valid everywhere.
Appendix~\ref{app:L2-certificate-audit} verifies the routine algebra:
\[
 D>0,\quad c>0,\quad \rho<2,\quad
 0<\tau<p_0<\lambda,\quad
 x_0:=1-\rho\tau\in(0,1),
\]
and \(E_{\rm pre}=0\), while the middle-suffix residual is tangent to
zero at \(p_0\).

Equation~\eqref{eq:prefix-residual-square}, with \(E_{\rm pre}=0\),
proves the prefix condition globally, with equality at \(x_0\).

For the suffix condition, the middle piece of
\eqref{eq:suffix-residual-sosa} is
\[
  \frac c2(p-p_0)^2.
\]
Since the residual is continuous at \(p=\tau>0\), the bracket in
the first piece has value \(c(\tau-p_0)^2/(2\tau)>0\) there.
This bracket is decreasing in \(p\), since \(\rho<2\), so it is
positive throughout \(0\leq p\leq\tau\). Multiplying by \(p\geq0\)
proves nonnegativity of the first piece.
At \(p=\lambda\), the high piece agrees with the middle square
in both value and first derivative. Its quadratic coefficient is
\(\alpha(\rho-1)\), so expansion about \(\lambda\) gives, for
\(p\geq\lambda\),
\[
 \alpha(\rho-1)(p-\lambda)^2
 +c(\lambda-p_0)(p-\lambda)
 +\frac c2(\lambda-p_0)^2,
\]
which is nonnegative because \(\rho>1\), \(c>0\), and
\(\lambda>p_0\). Thus every piece of
\eqref{eq:suffix-residual-sosa} is nonnegative.
\end{proof}

\begin{lemma}[The \(L_1\) certificate]\label{lem:L1-certificate}
Let \(1<\alpha\leq\alpha_0\), put \(s:=s(\alpha)\), and set
\[
 \rho:=L_1(\alpha)=\frac{s(2-s)}{2s-3},
 \qquad
 \tau:=\frac{2s-3}{(2-s)(s+1)},
 \qquad
 \lambda:=\frac{2s-3}{2-s}.
\]
Then \(0<\tau<\lambda\), and
\eqref{eq:prefix-certificate}--\eqref{eq:suffix-certificate} hold.
\end{lemma}

\begin{proof}
Here
\[
 x_0:=1-\rho\tau=\frac1{s+1}\in(0,1),
 \qquad
 p_1:=\frac1{2(\rho-1)}.
\]
These are precisely the prefix fraction and normalized large-job size in
\eqref{eq:lower-L1-contacts}; the same equation identifies
\(\tau,\lambda\) with the two loads of the normalized Case~I comparison
schedule.  We now verify that the certificate inequalities, which are
tight at \(x_0\) and \(p_1\), hold away from those contacts.
Appendix~\ref{app:L1-certificate-audit} verifies
\[
 0<\tau<\lambda\leq p_1,\qquad \rho<2,
\]
the identity \(E_{\rm pre}=0\), the high-suffix tangency at \(p_1\),
and the positive value of the first-region bracket at \(p=\tau\).

Equation~\eqref{eq:prefix-residual-square} therefore proves the prefix
condition globally, with equality at \(x_0\).

The high piece of \eqref{eq:suffix-residual-sosa} is
\[
  \alpha(\rho-1)(p-p_1)^2.
\]
This is nonnegative since \(\rho>1\).

The residual is continuous at both breakpoints. At \(p=\tau>0\),
the independently verified positive bracket value makes the first
piece, and hence the middle piece, positive. At \(p=\lambda\),
the middle piece equals the nonnegative high square. Its two
endpoint values are therefore nonnegative. If its second derivative
\(c=\rho(2\alpha-1)-2\alpha\) is nonpositive, concavity puts its minimum
at an endpoint. If \(c>0\), its derivative at \(\lambda\) agrees
with the high piece's derivative and is
\[
  2\alpha(\rho-1)(\lambda-p_1)\leq0.
\]
Since the derivative is increasing, the middle piece is nonincreasing
throughout \([\tau,\lambda]\), and its minimum is again its value at
\(\lambda\).

Finally, the bracket in the first piece is decreasing because
\(\rho<2\). Its positive value at \(p=\tau\) was established
independently in the appendix, so it stays positive throughout
\(0\leq p\leq\tau\). Multiplying by \(p\geq0\) proves nonnegativity
of the first piece. This proves the suffix condition for every \(p\geq0\).
\end{proof}

\begin{theorem}[Tight upper bound]\label{thm:tight-upper}
For every \(\alpha\geq1\),
\[
  \mathcal R(\alpha)\leq\mathcal L(\alpha).
\]
\end{theorem}

\begin{proof}
For \(\alpha=1\), the claim follows from
Remark~\ref{rem:kawaguchi-kyan}.  For \(\alpha\geq4\), the preliminary bound in
Theorem~\ref{thm:parallel-baseline} is \(\sqrt{\alpha}\).

It remains to consider \(1<\alpha<4\).  Fix \(m\).  By
Lemma~\ref{lem:equal-ratio-lifting}, it is enough to prove
the bound for common-ratio instances, and
Lemma~\ref{lem:endpoint-reduction} further reduces to
endpoint instances.  If
\(\alpha\leq\alpha_0\), apply Lemmas~\ref{lem:L1-certificate} and
\ref{lem:prefix-suffix} with \(\rho=L_1(\alpha)\).  If
\(\alpha>\alpha_0\), apply Lemmas~\ref{lem:L2-certificate} and
\ref{lem:prefix-suffix} with \(\rho=L_2(\alpha)\).  The resulting bound
holds for every \(m\), and therefore also after taking the supremum over
\(m\).
\end{proof}

\begin{proof}[Proof of Theorem~\ref{thm:intro-parallel}]
Theorem~\ref{thm:tight-lower} gives
\(\mathcal R(\alpha)\geq\mathcal L(\alpha)\), and
Theorem~\ref{thm:tight-upper} gives the reverse inequality.
Lemma~\ref{lem:remove-ties} shows that the lower bound is unchanged when
nominal Smith ratios are required to be distinct.
\end{proof}

\section{Algebraic certificates}
\label{app:algebraic-certificates}

This appendix records the exact algebra behind the two branches, in the
order in which it is used in the proof.  Thus every invocation of a
root equation or sign condition in the preceding proofs can be checked by
polynomial arithmetic.  We verified all polynomial identities, resultants,
and Sturm signs below using exact arithmetic. Supplementary scripts,
provided as ancillary files with this preprint, reproduce these checks;
they support the written algebra rather
than independently certify the scheduling and interval-continuity arguments.

\subsection{Deriving the branches from the active contacts}
\label{app:branch-derivation}

Write
\begin{align*}
 E_{\rm pre}
 &:=
 \rho\bigl(
 2\lambda-\lambda^2+2\beta\tau
 -\beta(\rho+1)\tau^2
 \bigr)-\alpha,\\
 E_{\rm mid}
 &:=
 (\alpha-\rho\lambda)^2-\rho\beta\tau^2c,\\
 E_{\rm high}
 &:=
 2\rho(\rho-1)(\lambda^2+\beta\tau^2)-\alpha,
\end{align*}
where \(c=\rho(2\alpha-1)-2\alpha\).  These are respectively the
prefix, middle-suffix, and high-suffix tangency residuals.

For the \(L_1\) contact pattern, introduce \(s\) through
\[
 \lambda=\frac{s}{\rho},
 \qquad
 \tau=\frac{s}{\rho(s+1)}.
\]
After multiplication by the common denominator \(\rho(s+1)^2\), the
equations \(E_{\rm pre}=E_{\rm high}=0\) have numerators
\begin{align*}
 N_{\rm pre}
 ={}&
 -\alpha\rho-\alpha s^2+2\rho s^3+3\rho s^2-s^4-2s^3,\\
 N_{\rm high}
 ={}&
 \alpha\rho s^2-2\alpha\rho s-\alpha\rho-2\alpha s^2\\
 &\quad
 +2\rho s^4+4\rho s^3-2s^4-4s^3.
\end{align*}
Because both numerators vanish,
\[
 0=N_{\rm high}-2N_{\rm pre}
 =\rho\bigl(\alpha(s-1)^2-2s^2(3-s^2)\bigr).
\]
Since \(\rho>0\), this gives
\[
 \alpha=\frac{2s^2(3-s^2)}{(s-1)^2}.
\]
Substituting this expression into \(N_{\rm pre}=0\) gives
\[
 \rho
 =\frac{s^4+2s^3+\alpha s^2}{2s^3+3s^2-\alpha}
 =\frac{s(2-s)}{2s-3}.
\]
Conversely, these two expressions make both numerators vanish.  Finally,
\[
 \operatorname{Res}_{s}
 \left(
  (2s-3)\rho-s(2-s),\
  \alpha(s-1)^2-2s^2(3-s^2)
 \right)
 =P_\alpha^{I}(\rho),
\]
which proves the claimed implicit equation without an unrecorded
elimination.

For the \(L_2\) contact pattern, the lower witness supplies
\[
 D=3\alpha\rho-\alpha-2\rho^2,\qquad
 \tau=\frac{2(\alpha-\rho)}D,\qquad
 \lambda=\frac{2\rho(\alpha-1)}D.
\]
Substitution, before imposing any root equation, gives the exact
factorizations
\[
 E_{\rm pre}=-\frac{P_\alpha^{II}(\rho)}{D^2},
 \qquad
 E_{\rm mid}=\frac{\alpha P_\alpha^{II}(\rho)}{D^2}.
\]
Consequently the prefix and middle-suffix contacts are simultaneously
active exactly on the branch \(P_\alpha^{II}(\rho)=0\); the interval in
\eqref{eq:P-II} selects the relevant root.

\subsection{The \texorpdfstring{\(L_2\)}{L2} root and feasible branch}
\label{app:L2-root-audit}

The endpoint evaluations are
\[
 P_\alpha^{II}(\sqrt\alpha)
 =\alpha^2(\sqrt\alpha-2)^2(\sqrt\alpha-1)^2,
\]
\[
 P_\alpha^{II}\!\left(\frac{3\alpha}{\alpha+2}\right)
 =
 \frac{4\alpha^3(\alpha-4)(\alpha-1)^2(4\alpha-7)}
      {(\alpha+2)^4}.
\]
For \(\alpha_0<\alpha<4\), the first expression is positive and the
second negative, so \(P_\alpha^{II}\) has a root in the defining interval.
The multiple-root resultant factors as
\[
 \operatorname{Res}_{r}
 (P_\alpha^{II},\partial_rP_\alpha^{II})
 =
 -1024\alpha^7(\alpha-1)^4(2\alpha-1)Q(\alpha),
\]
where
\[
 Q(\alpha)
 =10\alpha^4+595\alpha^3-2272\alpha^2
  +2688\alpha-1024.
\]
Its signs on the relevant interval follow successively from
\[
 Q(\alpha_0)=\frac{27959+5191\sqrt{33}}{648}>0,
\quad
 Q'(\alpha_0)=\frac{85279+13743\sqrt{33}}{216}>0,
\]
\[
 Q''(\alpha_0)=\frac{4691+1955\sqrt{33}}6>0,
\qquad
 Q'''(\alpha)=240\alpha+3570>0.
\]

At \(\alpha=2\), let
\begin{align*}
 f_0&=3r^4-10r^3+10r^2-4r+2,\\
 f_1&=(2r-1)(3r^2-6r+2),\\
 f_2&=15r^2-14r-14,\\
 f_3&=101-64r,\qquad f_4=-1.
\end{align*}
An exact Sturm computation gives the sequence
\[
 \left(
 \frac13f_0,\ \frac23f_1,\ \frac1{36}f_2,\
 \frac{16}{225}f_3,\ \frac{575}{16384}f_4
 \right);
\]
all five scaling factors are positive.  Its endpoint signs are
\[
\begin{array}{c|ccccc|c}
 r&f_0&f_1&f_2&f_3&f_4&\text{variations}\\ \hline
 \sqrt2&+&-&-&+&-&3\\
 3/2&-&-&-&+&-&2.
\end{array}
\]
These signs can be checked using rational arithmetic. For example,
\(f_0(\sqrt2)=34-24\sqrt2>0\) because \(34^2>2\cdot24^2\).
The other left-endpoint signs follow by the same comparison of squared
magnitudes, while the right-endpoint values are rational.
Hence \(P_2^{II}\) has one root in the defining interval.  The nonzero
multiple-root resultant and the two nonzero endpoint evaluations show
that this number cannot change for \(\alpha_0<\alpha<4\).  The selected
root is therefore unique and simple, and hence varies continuously with
\(\alpha\).  The same endpoint and resultant argument remains valid at
\(\alpha_0\), where direct substitution gives the unique root
\[
 L_1(\alpha_0)=\frac{17+\sqrt{33}}{16}
\]
in the closed defining interval.  This proves continuity at the branch
point.  As \(\alpha\uparrow4\), both endpoints of the defining interval
converge to \(2\), and therefore \(L_2(\alpha)\to2\).

For the feasible-branch sign, exact elimination gives
\[
 \operatorname{Res}_{r}(P_\alpha^{II},F_\alpha)
 =
 16\alpha^7(\alpha-4)(\alpha-1)^5
 (18\alpha^2-51\alpha+32).
\]
The roots of the final quadratic are
\((17\pm\sqrt{33})/12\), so the resultant does not vanish in the open
parameter interval.  At \(\alpha=2\),
\[
 P_2^{II}(29/20)=\frac{43}{40000}>0,\qquad
 P_2^{II}(73/50)=-\frac{88777}{1562500}<0,
\]
so the selected root lies in \((29/20,73/50)\).  Moreover,
\[
 F_2(r)=2r^3-12r^2+16r-4,\qquad
 F_2''(r)=12r-24<0
\]
on \((29/20,73/50)\), and
\[
 F_2'(29/20)=-\frac{1237}{200}<0,\qquad
 F_2(73/50)=\frac{317}{62500}>0.
\]
Thus \(F_2\) is decreasing and remains positive on this interval.  This
fixes the constant sign as positive and proves
\eqref{eq:case-II-feasible-branch}.

\subsection{The \texorpdfstring{\(L_2\)}{L2} lower witness}
\label{app:L2-lower-audit}

For compactness put
\[
 B:=3\alpha-(\alpha+2)\rho.
\]
The Case~II parameters in the lower proof satisfy the identities
\[
 1-\xi
 =\frac{2(\rho^2-\alpha)}
        {(2\rho-\alpha)(\rho-1)},\qquad
 1-u=\frac{B}{2(\alpha-\rho)},
\]
\[
 \zeta-1
 =\frac{2F_\alpha(\rho)}
 {\alpha(2\rho-\alpha)(\rho-1)^2}.
\]
They make the comparison exact because
\[
 A_\alpha(x,y,z)-\rho C_\alpha^{II}(x,y,z)
 =
 \frac{B\,P_\alpha^{II}(\rho)}
 {8(\alpha-\rho)^2(2\rho-\alpha)(\rho-1)}.
\]

Let \(\ell=x+z\), and normalize by \(\ell\) as in
Table~\ref{tab:lower-certificate}.  Direct cancellation gives
\[
 \frac{x}{\ell}=\frac{\alpha(\rho-1)}D,\qquad
 \frac{x}{u\ell}=\frac{2(\alpha-\rho)}D,\qquad
 \frac{\ell+y}{\ell}=\frac{2\rho(\alpha-1)}D.
\]
If \(p_L=1/\ell\) and
\[
 p_0:=\frac{\alpha-\rho\lambda}{c},
\]
then the final contact identity is exposed by
\[
 p_L-p_0
 =-\frac{P_\alpha^{II}(\rho)}{DBc}.
\]
Thus \(P_\alpha^{II}(\rho)=0\) gives
\eqref{eq:lower-L2-contacts}, including the active suffix equality
\(p_L=p_0\).

\subsection{Signs and contacts for the
  \texorpdfstring{\(L_2\)}{L2} certificate}
\label{app:L2-certificate-audit}

The defining interval and \(\alpha_0<\alpha<4\) imply, in order,
\[
 \rho>1,\qquad \rho^2-\alpha>0,\qquad
 \alpha-\rho>0,\qquad 2\rho-\alpha>0,\qquad B>0.
\]
Indeed,
\[
 \sqrt\alpha<\rho<\frac{3\alpha}{\alpha+2}<\min\{\alpha,2\},
\]
and \(2\sqrt\alpha>\alpha\) for \(\alpha<4\).  The denominator used by
the certificate has the positive decomposition
\[
 D=\alpha(\rho-1)+2\rho(\alpha-\rho)>0.
\]
Also
\[
 \alpha_0-\left(1+\frac{\sqrt3}{2}\right)
 =\frac{5+\sqrt{33}-6\sqrt3}{12}>0;
\]
indeed,
\((5+\sqrt{33})^2-(6\sqrt3)^2=10(\sqrt{33}-5)>0\).  Hence
\[
 c=(2\alpha-1)\rho-2\alpha
 >(2\alpha-1)\sqrt\alpha-2\alpha
 =\sqrt\alpha(2\alpha-2\sqrt\alpha-1)>0.
\]

The basic interval relations now follow without the root equation:
\[
 \lambda-\tau=\frac{2\alpha(\rho-1)}D>0,
\]
\[
 x_0=\frac{\alpha(\rho-1)}D>0,\qquad
 1-x_0=\frac{2\rho(\alpha-\rho)}D>0.
\]
For the middle contact,
\[
 p_0-\tau=\frac{\alpha B}{cD}>0
\]
and
\[
 \lambda-p_0=\frac{\alpha Q_{\rm int}}{cD},
 \qquad
 Q_{\rm int}:=
 4\alpha\rho^2-7\alpha\rho+\alpha-2\rho^2+4\rho.
\]
The exact polynomial identity
\[
 \alpha BQ_{\rm int}-2cF_\alpha(\rho)
 =-P_\alpha^{II}(\rho)
\]
shows that \(Q_{\rm int}>0\) on the selected root, by
\eqref{eq:case-II-feasible-branch}.  This proves
\[
 0<\tau<p_0<\lambda,\qquad 0<x_0<1.
\]
Finally, the two active tangencies are exactly
\[
 E_{\rm pre}=-\frac{P_\alpha^{II}(\rho)}{D^2},
 \qquad
 E_{\rm mid}=\frac{\alpha P_\alpha^{II}(\rho)}{D^2},
\]
and the first suffix piece at \(p=\tau\) has bracket
\[
 \alpha\left(\frac{\rho}{2}-1\right)\tau
 +\rho(\lambda+\beta\tau)-\alpha
 =\frac{\alpha(\rho^2-\alpha)}D>0.
\]
These are all sign and zero identities used in
Lemma~\ref{lem:L2-certificate}.

\subsection{Signs and contacts for the
  \texorpdfstring{\(L_1\)}{L1} certificate}
\label{app:L1-certificate-audit}

On \(s_0\leq s\leq s_+\), one has \(3/2<s<2\), so every denominator in
the \(L_1\) parametrization is positive and
\(\lambda=(s+1)\tau>\tau>0\).  The high contact satisfies
\[
 p_1-\lambda
 =
 \frac{(2s-3)(2s^2-s-4)}
 {2(2-s)(3-s^2)}
 \geq0,
\]
because \(s_0=(1+\sqrt{33})/4\) is the positive root of
\(2s^2-s-4\).

The bracket in the first suffix piece, evaluated at \(p=\tau\), is
\[
 \frac{s^2h(s)}{(s-1)^2(2-s)(s+1)},
 \qquad
 h(s):=s^4-s^3-s^2-5s+8,
\]
so it remains to show \(h(s)>0\).  Its
exact evaluation at the left endpoint gives
\[
 h(s_0)=\frac{259-45\sqrt{33}}{32}>0,\qquad
 h'(s_0)=\frac{-45+11\sqrt{33}}8>0,\qquad
 h''(s_0)=22.
\]
The first strict inequality follows since
\(259^2-(45\sqrt{33})^2=256>0\), and the second since
\((11\sqrt{33})^2-45^2=1968>0\).  Since
\[
 h'''(s)=24s-6>0,
\]
successive integration shows that \(h''\), \(h'\), and \(h\) remain
positive on the full interval.

The monotonicity direction in the first suffix piece uses
\[
 2-\rho=\frac{s^2+2s-6}{2s-3}>0.
\]
Here the numerator is increasing, and at \(s_0\) it equals
\[
 \frac{-27+5\sqrt{33}}8>0,
\]
because \((5\sqrt{33})^2-27^2=96>0\).  Finally, the active-equation
calculation in Appendix~\ref{app:branch-derivation} gives both
active-contact identities, using
\[
  \alpha(s-1)^2-2s^2(3-s^2)=0.
\]
This audits every sign and tangency used in
Lemma~\ref{lem:L1-certificate}.

\section{A two-job example}
\label{app:two-job-example}

On one machine, let \((\mu_1,\mu_2)=(1,1)\) and
\((w_1,w_2)=(3/2,1)\). The nominal Smith ratios are \(3/2\) and \(1\), so
\(\rwspt\) runs job 1 first. If the realized times are
\((p_1,p_2)=(3,1)\), satisfying the model with \(\alpha=3\), the jobs
complete at times \(3\) and \(4\), for cost \((3/2)\cdot3+1\cdot4=17/2\).
The true Smith ratios are instead \(1/2\) and \(1\). A clairvoyant
schedule, which knows the realized processing times in advance, runs job 2
first and has cost \(1\cdot1+(3/2)\cdot4=7\).
Thus this realization has ratio \(17/14>\rho_{\mathrm{KK}}\), exceeding
the classical parallel-machine bound even on one machine. The example
illustrates the loss from reversing the preferred order; it is not a
worst-case instance.

\paragraph*{Absolute cost versus approximation ratio.}
For these same nominal data, the realization with the largest algorithmic
cost need not have the largest approximation ratio. In the two-job example,
\((p_1,p_2)=(3,3)\) gives both \(\rwspt\) and the clairvoyant optimum
cost \(21/2\), hence ratio \(1\), whereas \((p_1,p_2)=(3,1)\) gives the
smaller algorithmic cost \(17/2\) but the larger ratio \(17/14\).

\section{Further related work}
\label{app:extended-related-work}

\paragraph*{Makespan with predictions.}
Prediction-error analyses also study makespan, the last job's completion
time, on uniform machines with different speeds~\cite{zhao2022uniform} and identical parallel
machines~\cite{bampis2026makespan,chen2026tighter}. The identical-machine
results also include nonpreemptive batches, with makespan as the objective.

\paragraph*{Interval uncertainty and regret.}
Pereira seeks a fixed single-machine sequence minimizing maximum absolute
regret in total weighted completion time~\cite{pereira2016robust}.
Sotskov, Egorova, and Lai study interval processing times under relative
regret~\cite{sotskov2009minimizing}: their objective is the maximum of
\((\operatorname{cost}(\pi,\bp)-\operatorname{OPT}(\bp))/
\operatorname{OPT}(\bp)\), which is the worst-case approximation ratio
of a fixed sequence minus one. They develop dominance relations,
two-job comparisons, and heuristics for this criterion.
Thus realization-specific comparison with OPT already appears in the
interval literature. Our theorems establish exact guarantees as a
function of common multiplicative distortion for the nominal Smith list
rule, extend the analysis to parallel machines, and transfer nominal
assignment quality to realized robustness.

\paragraph*{Queueing with estimated service times.}
There is also a complementary queueing literature on policies that use
inexact or predicted service times. Mean response time is average flow time.
Wierman and Nuyens quantify how inaccurate
size information affects size-based policies~\cite{wierman2008scheduling};
Mitzenmacher's price of misprediction compares mean response time under
predicted and perfect size information for the same policy
\cite{mitzenmacher2020scheduling}. Scully, Grosof, and Mitzenmacher prove
uniform approximation bounds relative to clairvoyant shortest remaining
processing time (SRPT) under two-sided multiplicative size estimates
\cite{scully2022uniform}. These papers use stochastic arrival and service
models; our guarantee concerns each realization of an adversarial batch.
Mitzenmacher and Vassilvitskii give a broader introduction to algorithms
with predictions~\cite{mitzenmacher2022algorithms}.

\end{document}